\documentclass[a4paper,10pt]{article}

\usepackage{amsmath}
\usepackage{amssymb}
\usepackage{latexsym}
\usepackage{ntheorem}
\usepackage{ifthen}
\usepackage{xcolor}
\usepackage{newcent}

\usepackage[margin=3cm]{geometry}

\theoremstyle{plain}

\theoremheaderfont{\bf}\theorembodyfont{\sl}

\newtheorem{theorem}{Theorem}[section]
\newtheorem{proposition}[theorem]{Proposition}
\newtheorem{lemma}[theorem]{Lemma}
\newtheorem{corollary}[theorem]{Corollary}

\theoremheaderfont{\it}\theorembodyfont{\rm}

\newtheorem{remark}[theorem]{Remark}

\theoremheaderfont{\bf}\theorembodyfont{\rm}

\newtheorem{problem}[theorem]{Problem}
\newtheorem{assumption}[theorem]{Assumption}

\theoremstyle{nonumberplain}

\theoremheaderfont{\bf}\theorembodyfont{\sl}

\newenvironment{proof}[1][]
{\ifthenelse{\equal{#1}{}}{\smallskip\noindent\textsl{Proof. }}{\smallskip
\noindent\textsl{Proof #1. }}}{\hfill$\Box$}

\newcommand{\esssup}{\displaystyle{\rm ess}\sup}

\begin{document}
\renewcommand{\thefootnote}{$*$}

\title{Maximization of Non-Concave Utility Functions in Discrete-Time Financial Market Models
\footnotemark}

\footnotetext{{\sl Abbreviated title:} Non-concave utility maximization in discrete-time
}

\author{Laurence Carassus\\ LMR, Universit\'e Reims Champagne-Ardenne \and Mikl\'os R\'asonyi\\ MTA Alfr\'ed R\'enyi Institute of Mathematics, Budapest}

\date{\today}

\maketitle

\begin{abstract}
This paper investigates the problem of maximizing
expected terminal utility in a (generically incomplete) discrete-time financial market
model with finite time horizon. In contrast to the standard setting,
a possibly non-concave utility function $U$ is considered, with domain of definition
$\mathbb{R}$. Simple conditions are presented which guarantee the
existence of an optimal strategy for the problem. In particular, the asymptotic
elasticity of $U$ plays a decisive role: existence can be shown when it is strictly greater at $-\infty$ than at
$+\infty$.
\end{abstract}

\section{Introduction}

The problem of maximizing expected utility is one of the most
significant issues in mathematical finance. To the best of our knowledge, the
first studies can be attributed to \cite{merton} and
\cite{samuelson}. In mathematical terms, $EU(X)$ needs to be
maximized in $X$, where $U$ is a concave increasing function and $X$
runs over values of admissible portfolios.  For general existence
results, we refer to \cite{RS05} in a discrete time setting and to
\cite{KS99} and \cite{S01} in continuous time models, see also
\cite{biagini-frittelli}, \cite{owen-zitkovic} and the references
therein for later developments.

Despite its ongoing success, the expected utility paradigm has been contested (see e.g. \cite{a53} and \cite{kt}). In particular, \cite{tk} suggested,
based on experimental evidence, that
the utility function should not be concave but rather ``$S$-shaped'',
i.e. $U(x)=U_+(x-B)$, $x\geq B $; $U(x)=-U_-(-(x-B))$, $x<B$ where $U_{\pm}:\mathbb{R}_+\to\mathbb{R}_+$ are
concave and increasing functions and $B\in\mathbb{R}$ is some reference point of the investor.

In this article we propose to consider a general, possibly non-concave utility function defined on the real line
(that can be ``$S$-shaped'' but our
results apply to a broader class of utility functions e.g. to piecewise concave ones).
As the objective function is non-concave, the mathematical treatment becomes difficult and only few related results
can be found in the literature.

Some authors have studied the rather specific case of
continuous-time complete markets (see \cite{cp} for piecewise
concave,  and \cite{bkp} for S-shaped utility functions or \cite{jz}
and \cite{cd}, where distortions on the objective probability are
considered) or one-period models (see \cite{bg} and \cite{hz}). See
also the recent paper of \cite{r12} in which utility maximisation is
carried out on the set of claims whose price is below a given constant for a fixed pricing measure.
Note that \cite{bkp}, \cite{cp},
\cite{cd} and \cite{r12} consider utility functions defined on the
positive half-line only, which leads to a considerably simpler
mathematical problem.

In the present article a general, generically incomplete, discrete-time
financial market model with finite horizon is considered together
with a possibly non-concave utility function $U$ defined on the real
line. In our recent paper \cite{cr11}, we study a similar framework
but with distortions on the objective probability. Under conditions
similar to Assumption \ref{ae} of the present paper,  a well-posedness result (i.e.
the objective function is finite) is established but the existence
of optimal strategies requires a particular structure for the information flow: the filtration should either be rich enough or there should exist an external source of randomness for the
strategies. In this setup it turns out, in contrast to the usual maximization of
expected utility problem, that an investor distorting the objective
probability may increase her satisfaction by exploiting randomized
trading strategies. So the existence result of \cite{cr11} is not
pertinent in the present setting without distortions and, to the
best of our knowledge, Theorem \ref{main}, Corollary
\ref{mortonhall}, Propositions \ref{camarche} and \ref{camarche1} below are the first
existence results for optimal portfolios maximizing expected
non-concave utility in an incomplete discrete-time model of a financial market.

The decisive sufficient conditions for existence are formulated below in
terms of the ``asymptotic elasticity'' of the function $U$ at
$\pm\infty$. This concept surged in the concave case, see
\cite{ck}, \cite{karatzas_et_al}, \cite{KS99} and \cite{S01}, which are the
early references. Let's denote by $u(x)$ the value function starting
from an initial wealth $x$. In \cite{KS99} it is showed, in a
general semimartingale model, that if $U$ (i) is strictly concave,
smooth and defined on $(0,+\infty)$, (ii) is such that there exists
$x$ satisfying $u(x)<\infty$ and (iii) has an asymptotic elasticity
at $+\infty$, called $AE_+(U$), strictly less than $1$, then an
optimal portfolio for the utility maximization problem
exists. If $U$ is defined over the whole real axis,
\cite{S01} showed existence assuming\footnote{A condition on the so-called dual optimizer is also
imposed and (ii) is replaced by the existence of some $x$ satisfying
$u(x)<U(\infty)$.}, in addition, that the
asymptotic elasticity of $U$ at $-\infty$, called $AE_-(U)$, is
strictly greater than $1$. This condition being close to
necessary (see section 3 of \cite{S01}), it has been generally
accepted as the standard assumption in continuous-time models, see
e.g. \cite{owen-zitkovic}. Note, however, that in a discrete-time
setting, when $U$ is concave and defined on $\mathbb{R}$, any of the two
assumptions $AE_+(U)<1$ or $AE_-(U)>1$ on its own is sufficient to
guarantee the existence of an optimal strategy (see \cite{RS05}).

In the present study a general continuous, increasing and possibly non-concave function $U$, defined on $\mathbb{R}$,
is considered and
we will assert the existence of an optimal strategy whenever $AE_+(U) <  AE_-(U)$,
where $AE_{\pm}(U)$ is an appropriate extension of the asymptotic elasticity concept
to non-differentiable  (and non-concave) functions. This generalizes results of \cite{RS05}.
Note that some conditions ensuring well-posedness  are also necessary to stipulate. We present easily verifiable
integrability assumptions to this end.

The key idea, as in \cite{RS05}, is to prove that strategies must satisfy
certain a priori bounds in order to be optimal and then one can use compactness arguments. A number of measure-theoretic
issues also need to be dealt with.

The paper is organized as follows: in section \ref{se2} we introduce our setup and state our main result; section
\ref{se3} establishes the existence of an optimal strategy for the one-step case. In section
\ref{dyn} we prove our main result, using dynamic programming, and provide easily verifiable sufficient
conditions for the market model that ensure well-posedness as well as the existence of an optimal strategy. Section \ref{se5} concludes, section \ref{qpp} collects some
useful measure-theoretic facts.

\section{Problem formulation}\label{se2}

Let $(\Omega,\Im, (\mathcal{F}_t)_{0\leq t\leq T},P)$ be a
discrete-time filtered probability space with time horizon
$T\in\mathbb{N}$. We assume that the sigma-algebras occurring in this
paper contain all $P$-zero sets.

Let $\{S_t,\ 0\leq t\leq T\}$ be a $d$-dimensional adapted process
representing the  price of $d$ risky securities in the
financial market in consideration.
There exists also a riskless
asset for which we assume a price constant $1$, for the sake of simplicity. Without this assumption, all the developments
below could be carried out using discounted prices.
The notation $\Delta
S_t:=S_t-S_{t-1}$ will often be used.
If $x,y\in\mathbb{R}^d$ then
the concatenation $xy$ stands for their scalar product. The symbol $|\cdot|$ denotes the Euclidean norm
on $\mathbb{R}^d$ (or on $\mathbb{R})$.

In what follows, $\Xi_t$ will denote the set of
$\mathcal{F}_t$-measurable $d$-dimensional random variables. Trading
strategies are represented by $d$-dimensional predictable
processes $(\phi_t)_{1\leq t\leq T}$, where $\phi_t^i$ denotes the
investor's holdings in asset $i$ at time $t$; predictability means
that $\phi_t\in\Xi_{t-1}$. The family of all predictable trading
strategies is denoted by $\Phi$.

From now on the positive (resp. negative) part of some number or
random variable $X$ is denoted by $X^+$ (resp. $X^-$). We will also write $f^{\pm}(X)$ for $\left(f(X)\right)^{\pm}$ for any random
variable $X$ and (possibly random) function $f$. We will
consider quasi-integrable random variables $X$, i.e. for any
sigma-field $\mathcal{H} \subset \Im$, $E(X|\mathcal{H})$ will be
defined by
$E(X|\mathcal{H})=E(X^+|\mathcal{H})-E(X^-|\mathcal{H})$, in a generalized sense, as soon as either $E(X^-|\mathcal{H})<\infty$ a.s. or
$E(X^+|\mathcal{H})<\infty$ a.s. This
implies that $E(X|\mathcal{H})$ can possibly be infinite. In
particular, $EX$ is defined (but can be infinity) whenever $EX^+$ or
$EX^-$ is finite. See section \ref{qpp} for more details on
generalized conditional expectations.

We assume that trading is self-financing. As the riskless asset's price is constant $1$, the value at time $t$ of a portfolio $\phi$ starting from
initial capital $x\in\mathbb{R}$ is given by
$$
V^{x,\phi}_t=x+\sum_{i=1}^t  \phi_i \Delta S_i.
$$

The following absence of arbitrage condition is standard, it is
equivalent to the existence of a risk-neutral measure in discrete-time markets with finite horizon, see e.g. \cite{dmw}.

\medskip

(NA) {\em If $V^{0,\phi}_T\geq 0$ a.s. for some $\phi \in\Phi$ then
$V^{0,\phi}_T=0$ a.s.}

\medskip

Let $D_t(\omega) \subset \mathbb{R}^d$ be the smallest affine
subspace containing the support of the (regular) conditional
distribution of $\Delta S_t$ with respect to $\mathcal{F}_{t-1}$,
i.e. $P(\Delta S_t \in \cdot\vert\mathcal{F}_{t-1})(\omega)$. Under (NA), it
is a non-empty $\mathcal{F}_{t-1}$-measurable random subspace, see Proposition \ref{randomD} below. If $D_t = \mathbb{R}^d$ then,
intuitively, there are no redundant assets. Otherwise, one may
always replace $\phi_t \in \Xi_{t-1}$ by its orthogonal projection
$\phi'_t$ on $D_t$ without changing the portfolio value since a.s.
$\phi_t \Delta S_t=\phi'_t \Delta S_t$, see Remark
\ref{haromcsillag} below as well as Remark 9.1 of \cite{fs}.

We will need a ``quantitative'' characterization of (NA).
From \cite{RS05} (see also \cite{JS98}), we know that:
\begin{proposition}\label{karakter}
(NA) implies the existence of $\mathcal{F}_t$-measurable random variables
$\delta_t,\kappa_t>0$
such that for all $\xi \in {\Xi}_t$ with  $\xi \in D_{t+1}$ a.s.:
\begin{equation}\label{valaki}
P( \xi \Delta S_{t+1} <-\delta_t |\xi| \vert\mathcal{F}_t)\geq \kappa_t
\end{equation}
holds almost surely; for all $0\leq t\leq T-1$.
\hfill $\Box$
\end{proposition}
\begin{remark}
The characterization of (NA) given by \eqref{valaki}
works only for $\xi \in D_{t+1}$ a.s. This is the reason why we will have to project
the strategy $\phi_{t+1} \in \Xi_{t}$  onto $D_{t+1}$ in our proofs. We refer again to Remark \ref{haromcsillag} below.
\end{remark}
We now present the conditions on $U$ which allow to assert the existence of an optimal strategy.
The main point here is that we do not assume concavity of $U$.

\begin{assumption}\label{ae}
The utility function $U:\mathbb{R}\to\mathbb{R}$ is non-decreasing, continuous and
$U(0)=0$. There
exist
$\underline{x}>0$, $\overline{x}>0$, $c\geq 0$, $\overline{\gamma}>0$ and $\underline{\gamma}>0$ such that
$\boxed{\overline{\gamma}<\underline{\gamma}}$ and for any
$\lambda\geq 1$,
\begin{eqnarray}
\label{ae+}
U(\lambda x) & \leq & \lambda^{\overline{\gamma}}U(x)+c \mbox{ for } x\geq \overline{x},\\
\label{ae-}
U(\lambda x) & \leq & \lambda^{\underline{\gamma}}U(x) \mbox{ for } x\leq -\underline{x},\\
\label{minus}
U(-\underline{x}) & < & 0. 
\end{eqnarray}
\end{assumption}

\begin{remark}
A typical example is given by $U(x)=\tilde{U}(x)-\tilde{U}(0)$, where
\begin{eqnarray*}
\tilde{U}(x) &=& \left\{
         \begin{array}{ll}
         U_+(x-B), & \;\;\; x \geq B \\
         -U_-(-x +B), & \;\;\; x < B,
         \end{array}
         \right.
\end{eqnarray*}
and $U_+(x)=a_+x^{\overline{\gamma}}$, $U_-(x)=a_-x^{\underline{\gamma}}$ with $a_{\pm}>0$, $B \in \mathbb{R}$ and
$0<\overline{\gamma}<\underline{\gamma}$.
\end{remark}
We could accommodate, at the price of more technical assumptions and complications, a random utility function.
This means that we could treat a random reference (benchmark) point $B$
as well and consider the problem of maximising $EU(V^{x,\phi}_T-B)$, but we refrain from doing so.
\begin{remark}
\label{remAE} In this remark, we comment on various items of
Assumption \ref{ae}. Fixing $U(0)=0$ is mere convenience. If $U$ is
strictly increasing then \eqref{minus} clearly follows from $U(0)=0$
and $\underline{x}>0$.

When $U$ is concave and differentiable, the ``asymptotic
elasticity'' of $U$ at $\pm \infty$ is defined as
\begin{eqnarray}\label{gonzalo1}
AE_+(U) & =& \limsup_{x\to\infty} \frac{U'(x)x}{U(x)},\\
AE_-(U)& = & \liminf_{x\to -\infty}\frac{U'(x)x}{U(x)},
\label{gonzalo2}
\end{eqnarray}
see \cite{KS99}, \cite{S01} and the references therein.

Assume for a moment that $c=0$.
It is shown in Lemma 6.3 of
\cite{KS99} that $AE_+(U)\leq \overline{\gamma}$ is equivalent to
\eqref{ae+}. Similarly, $AE_-(U)\geq\underline{\gamma}$ is
equivalent to \eqref{ae-}.
Note that the proof of Lemma 6.3 of
\cite{KS99} does not use the concavity of $U$. So if $U$ is
differentiable then \eqref{gonzalo1}, \eqref{gonzalo2} make sense and conditions \eqref{ae+} and \eqref{ae-} are equivalent
to $AE_+(U) \leq \overline{\gamma}$ and $\underline{\gamma} \leq
AE_-(U)$, respectively. It seems reasonable to extend the
definitions of $AE_+(U)$ (resp. $AE_-(U)$) to possibly
non-differentiable $U$ as the infimum (resp. supremum) of
$\overline{\gamma}$ (resp. $\underline{\gamma}$) such that
\eqref{ae+}  (resp. \eqref{ae-}) holds. Doing so we may see (looking
at Assumption \ref{ae}) that our paper asserts the existence of an
optimal strategy whenever there exist
$\overline{\gamma},\underline{\gamma}$ such that
\begin{eqnarray}
\label{otite}
AE_+(U) \leq \overline{\gamma}<\underline{\gamma} \leq AE_-(U).
\end{eqnarray}
The case $c>0$ is there only to handle bounded from above utility functions.
In the case of
a concave function $U$, it is easy to see that $U(\infty)<\infty$
implies that $AE_+(U)=0$ but this is not necessarily so for non-concave $U$.

Clearly, \eqref{otite} resembles the key condition in
\cite{S01}, namely $AE_+(U)<1<AE_-(U)$. Note that \cite{KS99}
requires only the condition $AE_+(U)<1$ since they are dealing with functions $U$ defined on
$(0,\infty)$. The condition of \cite{RS05}, in a discrete-time
setting like ours, is either $AE_+(U)<1$ or $1<AE_-(U)$.
When $U$ is concave, \eqref{ae+} and \eqref{ae-} always hold with $\overline{\gamma}=\underline{\gamma}=1$, i.e. $AE_+(U) \leq 1  \leq AE_-(U)$ (see Lemma 6.2 in \cite{KS99} and
the discussion after Definition 1.4 in \cite{S01}) so our paper generalizes \cite{RS05} to $U$ that is not necessarily concave.

We finish this remark with a comment on the condition $\overline{\gamma}<\underline{\gamma}$.  It is, in some sense, minimal as one can see from
the following example. Assume that
\begin{eqnarray*}
{U}(x) &=& \left\{
         \begin{array}{ll}
        x^{\alpha}, & \;\;\; x \geq 0 \\
         -(-x)^{\beta}, & \;\;\; x < 0,
         \end{array}
         \right.
\end{eqnarray*}
with
$\alpha \geq \beta$. Here one has $\overline{\gamma}=\alpha$ and $\underline{\gamma}=\beta$. Assume that
$S_0=0$, $\Delta S_1=\pm 1$ with
probabilities $p,1-p$ for some $0<p<1$. Then one gets
$$E(U(0+n\Delta S_1))=p n^{\alpha} -(1-p) n^{\beta}.$$
If $\alpha> \beta$, choose $p=1/2$ and $E(U(n\Delta S_1))$ goes to $\infty$ as $n\to\infty$.
If $\alpha= \beta$, choose $p>1/2$ and $E(U(n\Delta S_1))=n^{\alpha}(2p -1)$ goes to $\infty$ again as $n\to\infty$.
So in the case $\overline{\gamma}\geq \underline{\gamma}$ the given problem immediately becomes ill-posed,
even in this very simple example.
\end{remark}

\begin{remark} As it becomes clear from the proof,
one could weaken \eqref{ae+} and \eqref{ae-} in Assumption \ref{ae} above to \eqref{egyedik} and \eqref{kettedik} below. These latter conditions, however,
seem to be only marginally weaker than \eqref{ae+}, \eqref{ae-} and they
lack a natural mathematical or economical interpretation while \eqref{ae+} and \eqref{ae-} show a nice consistency
with the well-studied concave case, as pointed out in the previous Remark.
\end{remark}

\begin{problem}\label{leprobleme}
In this paper, we are dealing with maximizing the expected terminal utility
$EU(V^{x,\phi}_T)$ from initial endowment $x$. Namely, we consider
\begin{eqnarray*}
u(x)=\sup_{\phi\in\Phi(U,x)}EU(V^{x,\phi}_T),
\end{eqnarray*}
where $\Phi(U,x)$ is the set of strategies $\phi\in\Phi$
for which
$E[U(V^{x,\phi}_T)]$ exists and is finite.
\end{problem}

\begin{remark}
In \cite{S01} the existence of optimal strategies is investigated on some
enlargement of the class of strategies with $V^{x,\phi}_t$ bounded from below.
In a discrete time setup such constraints are not suitable. For example, if $T=1$ and $\Delta S_1$ follows the standard Gaussian law then only the strategy $\phi=0$
leads to $V^{x,\phi}_1$ bounded from below.
So here we choose to work on a much larger class,
where we only require that $E[U(V^{x,\phi}_T)]$ exists and  is finite. We will see that the price to pay is in terms of integrability:
without further assumptions our candidate for optimal solution $\phi^*$ will not necessarily stay in this class,
see the formulation of Theorem \ref{main} below.
\end{remark}

We will use a dynamic programming procedure
and, to this end, we have to prove that the associated random functions are well defined  and a.s. finite
under appropriate integrability conditions. Namely we prove in Proposition \ref{propre1} that if $U:\mathbb{R}\to\mathbb{R}$ is non-decreasing and left-continuous and if we
assume that for all $1\leq t\leq T$, $x \in \mathbb{R}$ and $y\in\mathbb{R}^d$
\begin{eqnarray*}
E(U^-(x+y\Delta S_{t})|\mathcal{F}_{t-1}) <  +\infty
\end{eqnarray*}
holds true a.s., then the following random functions are well-defined recursively, for all $x\in\mathbb{R}$ (we omit dependence on
$\omega\in\Omega$ in the notation):
\begin{eqnarray}\label{tuskes}
U_T(x) &:= & U(x),\\
U_{t-1}(x) & := & \esssup_{\xi\in \Xi_{t-1}}E(U_{t}(x+ \xi\Delta S_{t})\vert\mathcal{F}_{t-1})
\mbox{ a.s.,  for } 1\leq t \leq T,\label{vanek}
\end{eqnarray}
and one can choose $(-\infty,+\infty]$-valued versions which are a.s. non-decreasing and left-continuous (in $x$).

In order to have a
well-posed problem, we impose Assumption \ref{bellmann} below.
\begin{assumption}\label{bellmann}
For all $1\leq t\leq T$, $x \in \mathbb{R}$ and $y\in\mathbb{R}^d$
we assume that
\begin{eqnarray}
\label{negaU0}
E(U^-(x+y\Delta S_{t})|\mathcal{F}_{t-1})& < & +\infty\mbox{ a.s.}\\
\label{lasagna}
EU_0(x) & < & + \infty.
\end{eqnarray}
\end{assumption}
Note that by Proposition \ref{propre1},
one can state \eqref{lasagna}: $U_0$ is well defined under \eqref{negaU0} assuming only that $U$ is non-decreasing and continuous.
\begin{remark}
\label{rembellman}
\rm{In Assumption \ref{bellmann}, condition \eqref{lasagna} is not easy to verify.
We propose in Proposition \ref{camarche} a fairly general setup where it is satisfied,
see also Corollary \ref{mortonhall} and Proposition \ref{camarche1}.
In contrast, \eqref{negaU0} is a straightforward integrability condition on $S$. For instance, if
$U(x) \geq -m(1+|x|^p)$ for some $p,m>0$ and  $E|\Delta S_{t}|^p<\infty$ for all $t\geq 1$ then \eqref{negaU0} holds.}
\end{remark}

We are now able to state our main result.
\begin{theorem}\label{main}
Let $U$ satisfy Assumption \ref{ae}
and $S$ satisfy the (NA) condition. Let Assumption \ref{bellmann} hold.
Then one can choose non-decreasing, continuous in $x\in\mathbb{R}$ and $\mathcal{F}_t$-measurable
in $\omega\in\Omega$ versions of the random functions $U_t$ defined in \eqref{tuskes} and \eqref{vanek}.
Furthermore, there exists a ``one-step optimal'' strategy $\tilde{\xi}_t(x)\in\Phi$ satisfying, for all $t=1,\ldots,T$, and for each $x\in\mathbb{R}$,
$$
E(U_{t}(x+\tilde{\xi}_t(x) \Delta S_{t})\vert\mathcal{F}_{t-1})=U_{t-1}(x)\mbox{ a.s.}
$$
Using these $\tilde{\xi}_{\cdot}(\cdot)$, we define recursively:
$$
\phi^*_1:=\tilde{\xi}_1(x),\quad \phi^*_{t}:=
\tilde{\xi}_{t}\left(x+\sum_{j=1}^{t-1} \phi^*_j \Delta S_j\right),\
1\leq t\leq T.
$$
If, furthermore, $EU(V^{x,\phi^*}_T)$ exists then $\phi^*\in\Phi(U,x)$ and
$\phi^*$ is a solution of Problem \ref{leprobleme}.
\end{theorem}

We present the proof of Theorem \ref{main}
in section \ref{dyn}. To demonstrate its applicability, we state a simple corollary below.
Later we will also provide a quite general setup where Theorem \ref{main} applies and where
$EU(V^{x,\phi^*}_T)$ can be shown to exist (see Propositions \ref{camarche} and \ref{camarche1} in section \ref{dyn}).

\begin{corollary}\label{mortonhall}
Assume that (NA) holds and
the utility function $U:\mathbb{R}\to\mathbb{R}$ is strictly increasing, continuous, bounded from above with
$U(0)=0$ and satisfies \eqref{negaU0}. Assume also that there
exist
$\underline{x}>0$ and $\underline{\gamma}>0$ such that for any $\lambda\geq 1$,
$
U(\lambda x)  \leq  \lambda^{\underline{\gamma}}U(x) \mbox{ for } x\leq -\underline{x}.
$
Then defining $\phi^*$ as in Theorem \ref{main}, we get that $\phi^*\in\Phi(U,x)$ and
$\phi^*$ is a solution of Problem \ref{leprobleme}.
\end{corollary}
\begin{proof}
As $U$ is bounded from above, \eqref{lasagna} and thus Assumption \ref{bellmann} trivially holds. So do \eqref{minus} and
\eqref{ae+} (with, say, $\overline{\gamma}:=\underline{\gamma}/2$, $\overline{x}:=1$ and $c$ any positive upper
bound for $U(\infty)$).
Hence Assumption
\ref{ae} is true. Since $U$ is bounded from above, $E[U(V^{x,\phi^*}_T)]$ exists
automatically. Now Corollary \ref{mortonhall} follows from Theorem \ref{main}.
\end{proof}

\begin{remark}
In the absence of a concavity assumption on $U$ we cannot expect to have a \emph{unique}
optimal strategy.
\end{remark}

\section{Existence of an optimal strategy for the one-step case}\label{se3}
\label{onestep}
First we prove the existence of an optimal strategy in the case of a one-step model. To this aim we introduce (i)
a random function $V$, (ii) two $\sigma$-algebras $\mathcal{H}\subset\mathcal{F}$ containing $P$-zero sets,
(iii) a $d$-dimensional $\mathcal{F}$-measurable random variable  $Y$.

Let $\Xi$ denote the family of $\mathcal{H}$-measurable $d$-dimensional random variables.
The aim of this section is to study $\mathrm{ess.}\sup_{\xi\in \Xi}E(V(x+\xi Y)\vert\mathcal{H})$. For each $x$, let us fix an
arbitrary version $v(x)=v(\omega,x)$ of this essential supremum.

We prove in Proposition \ref{candy} that, under suitable assumptions, there is an optimiser $\tilde{\xi}(x)$
which attains the essential supremum in the definition of $v(x)$, i.e.
\begin{eqnarray}\label{infune}
v(x)= E(V(x+\tilde{\xi}(x) Y)\vert\mathcal{H}).
\end{eqnarray}

In Proposition \ref{candy}, we even prove that the same optimal solution $\tilde{\xi}(H)$ applies if we replace $x$ by any scalar
$\mathcal{H}$-measurable random variable $H$ in \eqref{infune}.

This setting will be applied in section \ref{dyn} with the choice
$\mathcal{H}=\mathcal{F}_{t-1}, \mathcal{F}=\mathcal{F}_{t},  Y=\Delta S_t$; $V(x)$
will be the maximal conditional expected
utility from capital $x$ if trading begins at time $t$, i.e. $V=U_t$. In this case, the function $v(x)$ will represent the maximal expected utility from capital $x$ if trading begins at time
$t-1$.

We start with a useful Lemma.
\begin{lemma}\label{determ}
Let $V(\omega,x)$ be a function from
$\Omega \times \mathbb{R}$ to $[-\infty,\infty]$ such that for almost
all $\omega$, $V(\omega,\cdot)$ is a nondecreasing
function.
The following conditions are equivalent~:
\begin{enumerate}
\item $E(V^+(x+ yY)\vert\mathcal{H})
 <  +\infty$ a.s., for all $x\in\mathbb{R}$, $y\in\mathbb{R}^d$.
\item $E(V^+(x+|y| |Y|)\vert\mathcal{H})
 <  +\infty$ a.s., for all $x,y\in\mathbb{R}$.
 \item $E(V^+(H+\xi Y)\vert\mathcal{H})
 <  +\infty$ a.s., for all $H, \xi$ $\mathcal{H}$-measurable random variables ($H$ is one-dimensional and $\xi$ is $d$-dimensional).
\end{enumerate}
The following conditions are equivalent~:
\begin{enumerate}
\item $E(V^-(x+ yY)\vert\mathcal{H})
 <  +\infty$ a.s., for all $x\in\mathbb{R}$, $y\in\mathbb{R}^d$.
\item $E(V^-(x-|y| |Y|)\vert\mathcal{H})  <  +\infty$ a.s., for all $x,y\in\mathbb{R}$.
 \item $E(V^-(H+\xi Y)\vert\mathcal{H})
 <  +\infty$ a.s., for all $H, \xi$ $\mathcal{H}$-measurable random variables ($H$ is one-dimensional and $\xi$ is $d$-dimensional).
\end{enumerate}
\end{lemma}
\begin{proof}
We only prove the equivalences for $V^+$ since the ones for $V^-$ are similar.
We start with 1. implies 2.
Introduce the following vectors for each function $i\in W:=\{-1,+1\}^d$:
\begin{eqnarray}
\label{signe}
\theta_i:=(i(1)\sqrt{d},\ldots,i(d)\sqrt{d}).
\end{eqnarray}
Let $x,y\in\mathbb{R}$. We can conclude since
$$
V^+(x +|y| |Y|)\leq \max_{i\in W} V^+(x+ |y| \theta_i Y)\leq \sum_{i\in W} V^+(x+ |y| \theta_i Y),
$$
by $|Y|\leq\sqrt{d}(|Y^1|+\ldots+|Y^d|)$.
Next we prove that 2. implies 3. Let $H, \xi$ be $\mathcal{H}$-measurable random variables,
define $A_m:=\{|H|<m,|\xi|<m\}$
for $m\geq 1$ and $Z:=E(V^+(H+\xi Y)\vert\mathcal{H})$. Then
$E(Z1_{A_m}|\mathcal{H}) \leq 1_{A_m}E(V^+(m+m|Y|)|\mathcal{H})$ and the latter exists and it is finite by 2. Hence we can conclude by Corollary \ref{nagyy}.
Now 3. trivially implies 1.
\end{proof}

A first step consists in showing that, under weak assumptions, one can choose a $(-\infty,+\infty]$-valued version  of $v(x)$
which is a.s. non-decreasing and left-continuous (in $x$). This will allow us later to prove Proposition
\ref{propre1}, i.e. that one can choose $(-\infty,+\infty]$-valued versions of the random functions $U_t$ which are a.s. non-decreasing and left-continuous (in $x$).

\begin{lemma}
\label{propre}
Let $V(\omega,x)$ be a function from
$\Omega \times \mathbb{R}$ to $(-\infty,\infty]$ such that for almost
all $\omega$, $V(\omega,\cdot)$ is a nondecreasing, left-continuous
function and
$V(\cdot,x)$ is $\mathcal{F}$-measurable for each fixed $x$.
Assume that, for all $1\leq t\leq T$, $x \in \mathbb{R}$ and $y\in\mathbb{R}^d$,
\begin{eqnarray}
\label{touxV}
E(V^-(x+yY)|\mathcal{H}) <  +\infty
\end{eqnarray}
holds true a.s. Then one can choose for all $x\in\mathbb{R}$
a $(-\infty,+\infty]$-valued version of $v(x)$ which is a.s. non-decreasing and left-continuous (in $x$). In particular, this version of $v$
is $\mathcal{H}\otimes\mathcal{B}(\mathbb{R})$-measurable.
\end{lemma}
\begin{proof} First, by Lemma \ref{determ}, \eqref{touxV} implies
$
E(V^-(x+\xi Y)|\mathcal{H}) <  +\infty
$
a.s. for $\xi\in\Xi$ as well.
For $x\in\mathbb{R}$, let $\mathfrak{v}(x)$ be an
arbitrary version of $\esssup_{\xi\in \Xi} E(V(x+\xi Y)|\mathcal{H})$.
Fix any pairs of real numbers
$x_1\leq x_2$. As for almost
all $\omega$, $V(\omega,\cdot)$ is a nondecreasing, we get on full measure set that for all $\xi \in \Xi$, $V(x_1+\xi Y) \leq V(x_2+\xi Y)$.
By monotonicity of the conditional expectations and the essential supremum, we obtain that
$\mathfrak{v}(x_1)\leq \mathfrak{v}(x_2)$ almost surely. Hence there is a
negligible set $N\subset\Omega$ outside which $\mathfrak{v}(\omega,\cdot)$ is non-decreasing over $\mathbb{Q}$. Note that here $N \in \mathcal{H}$ since $\mathcal{H}$
contains $P$-zero sets by assumption.

For $\omega\in\Omega\setminus N$, let us define  the following left-continuous function on $\mathbb{R}$ (possibly taking the value $\infty$):
for each $x\in\mathbb{R}$ let
$\mathfrak{A}(\omega, x):=\sup_{r<x,r\in\mathbb{Q}}\mathfrak{v}(\omega, r)$. For $\omega\in N$, define $\mathfrak{A}(\omega,x)= 0$ for all $x\in\mathbb{R}$.
Let $r_i$, $i\in\mathbb{N}$ be an enumeration of $\mathbb{Q}$. Then
$\mathfrak{A}(\omega,x)=\sup_{n\in\mathbb{N}} [\mathfrak{v}(\omega,r_n)1_{\{r_n<x\}}+(-\infty)1_{\{r_n\geq x\}}]$ for all $x$ and for all $\omega\in\Omega\setminus N$, hence $\mathfrak{A}$
is clearly an
$\mathcal{H}\otimes\mathcal{B}(\mathbb{R})$-measurable function.

It remains to show that, for each fixed $x\in\mathbb{R}$, $\mathfrak{A}(x)$ is a version of $v(x)$.
It suffices to show that, for each $x$, $\mathfrak{A}(x)$ is equal to $\mathfrak{v}(x)$ almost surely (where the zero-set may depend
on $x$) since, $\mathfrak{v}(x)$ being a version of the essential supremum, so will be $\mathfrak{A}(x)$, too.

Take increasing rationals $r_n \uparrow x$, $r_n<x$, $n\to\infty$. Then $\mathfrak{v}(r_n) \leq \mathfrak{v}(x)$ a.s. and
$\mathfrak{A}(x)=\lim_n \mathfrak{v}(r_n)\leq \mathfrak{v}(x)$
a.s. On the other hand, for each $k\geq 1$, we claim that there is $\xi_k\in\Xi$ such that
$$
\mathfrak{v}(x)-1/k=\esssup_{\xi\in \Xi}E(V(x+ \xi Y)\vert\mathcal{H})-1/k\leq
E(V(x+ \xi_k Y)\vert\mathcal{H})
$$ a.s.
Indeed, as $E(V(x+ \xi Y)\vert\mathcal{H})$, $\xi\in\Xi$ is easily seen to be directed upwards,
there is a sequence $\zeta_n\in\Xi$ such that $E(V(x+ \zeta_n Y)\vert\mathcal{H})$ is a.s. nondecreasing
and converges a.s. to $\mathfrak{v}(x)$. We can define $\xi_k:=\zeta_{l(k)}$ where
$l(k)(\omega):=\inf\{l:E(V(x+ \zeta_l Y)\vert\mathcal{H})(\omega)\geq \mathfrak{v}(\omega,x)-1/k\}$.

By definition, $\mathfrak{v}(r_n)\geq E(V(r_n+ \xi_k Y)\vert\mathcal{H})$ a.s. for all $n$.
We argue over the sets $A_m(k):=\{\omega: m-1\leq |\xi_k(\omega)|<m\}$, $m\geq 1$ separately and fix $m$.
Provided that we can apply Fatou's lemma, we get
$$
\mathfrak{A}(x)=\lim_n \mathfrak{v}(r_n)=\liminf_n \mathfrak{v}(r_n)\geq E(V(x+ \xi_k Y)\vert\mathcal{H})\mbox{ a.s. on }A_m(k),
$$
using left-continuity of $V$. It follows that $\mathfrak{A}(x)\geq \mathfrak{v}(x)-1/k$ a.s. for all $k$, hence
$\mathfrak{A}(x)\geq \mathfrak{v}(x)$ a.s. So necessarily $\mathfrak{A}(x)=\mathfrak{v}(x)$ a.s. and $\mathfrak{A}$ is a suitable
version, as claimed. This also implies that $\mathfrak{A}$ is a.s. decreasing as $\mathfrak{v}$ is.

Fatou's lemma works because of \eqref{touxV} and the estimate
$$
V^-(x+\xi_k Y)\leq \max_{i\in W} V^-(x- m\theta_i Y)\leq \sum_{i\in W} V^-(x- m\theta_i Y)\mbox{ a.s.},
$$
which holds on $A_m(k)$, for each $m,k$ (see \eqref{signe} for the definition of $\theta_i$).
\end{proof}

Now we introduce the random set $D$ such that
for all $\omega \in \Omega$, $D(\omega)$ is
the smallest affine subspace  containing the support of the conditional
distribution of $Y$ with respect to $\mathcal{H}$, i.e. $P(Y \in \cdot\vert\mathcal{H})(\omega)$.

In order to prove  \eqref{infune}, we impose the following conditions on $D$, $Y$, $V$ and ${\cal H}$:
\begin{assumption}
\label{hypD}
We have $D \in {\cal B}(\mathbb{R}^d) \otimes {\cal H}$ and for almost all $\omega$,
$D(\omega)$ is a non-empty vector subspace of $\mathbb{R}^d$.
\end{assumption}

\begin{remark}\label{haromcsillag}
Let $\xi\in\Xi$ and let $\xi'\in\Xi$  be the orthogonal
projection of $\xi$ on $D$ (this is $\mathcal{H}$-measurable by Proposition 4.6 of \cite{RS05}). Then $\xi-\xi'\perp D$ a.s. hence
$\{Y\in D\}\subset\{(\xi-\xi')Y=0\}$. It follows that
$$
P(\xi Y=\xi' Y\vert\mathcal{H})=P((\xi-\xi')Y=0\vert\mathcal{H})\geq
P(Y\in D\vert\mathcal{H})=1
$$
a.s., by the definition of $D$. Hence $P(\xi Y=\xi' Y)=E(P(\xi Y=\xi' Y\vert\mathcal{H}))=1$.
\end{remark}

\begin{assumption}
\label{AOAone}
There exist $\mathcal{H}$-measurable random variables with $0<\alpha, \beta\leq 1$ a.s.
such that for all $\xi\in \Xi$
with $\xi \in D$ a.s.:
\begin{eqnarray}
\label{lim}
P( \xi Y \leq -\alpha |\xi|\vert\mathcal{H}) \geq \beta.
\end{eqnarray}
\end{assumption}

\begin{assumption}
\label{hypV1}
$V(\omega,x)$ is a function from
$\Omega \times \mathbb{R}$ to $\mathbb{R}$ such that for almost
all $\omega$, $V(\omega,\cdot)$ is a nondecreasing, finite-valued, continuous
function and
$V(\cdot,x)$ is $\mathcal{F}$-measurable for each fixed $x$.
\end{assumption}

We also need the following integrability conditions:
\begin{assumption}
\label{hypVmoinsinfty}
For all $x,y\in\mathbb{R}$,
\begin{eqnarray}\label{veges}
E(V^-(x-|y| |Y|)\vert\mathcal{H}) & < & +\infty\quad\mathrm{a.s.} \\
\label{nemvegtelen}
E(V^+(x+|y| |Y|)\vert\mathcal{H})
 & < & +\infty\quad\mathrm{a.s.}.
\end{eqnarray}
\end{assumption}

\begin{remark}\label{dunoon}
Let $H,\xi$ be arbitrary $\mathcal{H}$-measurable random variables. Then, from Lemma \ref{determ}, under Assumption
\ref{hypVmoinsinfty} above, $E(V(H+\xi Y)\vert \mathcal{H})$ exists and it is a.s. finite.
\end{remark}

We finally assume the following growth conditions on $V$.
\begin{assumption}
\label{hypV2}
There exists some constants $C\geq 0$, $\underline{\gamma}>\overline{\gamma}>0$
such that, outside a fixed negligible set,
\begin{eqnarray}\label{bajka}
V(\lambda x) & \leq & \lambda^{\overline{\gamma}} V(x)+C\lambda^{\overline{\gamma}},\\
\label{bajka-}
V(\lambda x) & \leq & \lambda^{\underline{\gamma}} V(x)+C\lambda^{\underline{\gamma}}
\end{eqnarray}
hold for all $x\in\mathbb{R}$ and $\lambda\geq 1$.
\end{assumption}
\begin{assumption}
\label{hypV3}
There exists a non-negative, $\mathcal{H}$-measurable, a.s. finite valued random variable $N$ such that
\begin{equation}\label{limproba}
P\left(V(-N)<-\frac{2C}{\beta}-1 \vert\mathcal{H}\right) \geq 1-\beta/2 \quad\mathrm{a.s.}
\end{equation}
for $\beta$ is defined in Assumption \ref{AOAone} and $C$ in Assumption \ref{hypV2}.
\end{assumption}

We
briefly sketch the strategy for proving the existence of an optimiser $\tilde{\xi}(x)$
which attains the essential supremum in the definition of $v(x)$ (see \eqref{infune}).
First, we prove that strategies, in order to be optimal, have to be
bounded  by  some random variable $\tilde{K}$ (Lemmata \ref{Kbound} and \ref{bound}).
Then we establish that $E(V(x+yY) \vert {\cal H})$ has a version $G(\omega,x,y)$
which is jointly continuous in $(x,y)$ with probability $1$ (Lemma \ref{candycontinue}).

\noindent Let $A^{\tilde{K}}(\omega,x)=\sup_{y \in \mathbb{Q}^d,|y|\leq \tilde{K}(x)}G(\omega,x,y)$.
We prove that $A^{\tilde{K}}$ is continuous in $x$ and  that $A=A^{\tilde{K}}$ outside a negligible set, where
$A(\omega,x)=\sup_{y \in \mathbb{Q}^d}G(\omega,x,y)$ (Lemma \ref{candymaxiessup}).
Furthermore, we show for each $x$ that $v(x)=A(x)$ a.s. hence $A(\cdot)$ is an almost surely
continuous version of the essential supremum $v(\cdot)$. Based on the preceding steps, we
can  construct a
sequence $\xi_n(\omega,x)$ taking values in $D$ along which the supremum
in the definition of the function $A$ is attained and $\xi_n$ is
also jointly measurable (Lemma \ref{candymaxi}).
The bound $\tilde{K}$ and a compactness argument provide a limit $\tilde{\xi}$ of $\xi_n$
(Proposition \ref{candy}), which turns out to be the optimiser in \eqref{infune}.

\begin{lemma}\label{Kbound}
Let Assumptions \ref{hypD}, \ref{AOAone}, \ref{hypV1}, \ref{hypVmoinsinfty}, \ref{hypV2} and \ref{hypV3} hold.
Let $\eta$ such that $0<\eta<1$ and
$\overline{\gamma}<\eta \underline{\gamma}$ (recall that $\overline{\gamma}<\underline{\gamma}$).
Let $x, y \in \mathbb{R}$ with $x<y$. Define
\begin{eqnarray}
\label{duracell}
L & = & E(V^+(1 +|Y|)\vert {\cal H}),\\
\label{K1}
K_1 (x) & =  & \max \left(1,x^+, \left(\frac{x^+ +N}{\alpha}\right)^{\frac1{1-\eta}}, \frac{x^+ +N}{\alpha},\left(\frac{6L}{\beta}\right)^{\frac1{\eta \underline{\gamma}-\overline{\gamma}}},
\left(\frac{6C}{\beta}\right)^{\frac1{\eta \underline{\gamma}-\overline{\gamma}}}\right),\\
\label{K2}
K_2 (x) & = & \left(\frac{6[E(V(-x^-)|\mathcal{H})]^-}{\beta}\right)^{\frac1{\eta \underline{\gamma}}},\\
\label{K}
{K}(x,y) & = & \max(K_1 (y),K_2 (x)),\\
\label{Ktilde}
\tilde{K}(x) & = & K(\lfloor x\rfloor,\lfloor x\rfloor+1),
\end{eqnarray}
where $\lfloor x\rfloor$ denote the largest integer $n$ with $n\leq x$.
Then all these random variables are $\mathcal{H}$-measurable and a.s. finite-valued.
$K_1(\omega,x)$ (resp. $K_2(\omega,x)$) is non-decreasing (resp. non-increasing) in $x$.
The random function $\tilde{K}(\cdot)$
is  $\mathcal{H}\otimes\mathcal{B}(\mathbb{R})$-measurable and a.s. constant on intervals of the form $[n,n+1)$, $n\in\mathbb{Z}$.

\noindent For $\xi\in \Xi$
with $\xi \in D$ a.s. and $|\xi|\geq \tilde{K}(x)$, we have almost surely:
\begin{eqnarray}
\label{petitjoli}
E(V(x+ \xi Y)\vert {\cal H}) \leq E(V(x)\vert {\cal H}).
\end{eqnarray}
Assume that there exist numbers $m,p>0$ such that $V(x)\geq -m(1+|x|^p)$ a.s. for
all $x \leq 0$. Then there exists a  non-negative, a.s. finite-valued
$\mathcal{H}$-measurable random variable $M$  and some number
$\theta>0$ such that, for a.e. $\omega$,
\begin{eqnarray}
\label{stratborne}
\tilde{K}(x) & \leq &  M (|x|^{\theta}+1),\mbox{ for all $x$},
\end{eqnarray}
and $M$ is a polynomial function of
$N,1/\alpha,1/\beta$ and $L$.
\end{lemma}
It follows directly from \eqref{petitjoli} that $E(V(x+ \xi 1_{|\xi|> \tilde{K}(x)} Y)\vert {\cal H}) \leq E(V(x)\vert {\cal H})$ a.s. for all $\xi\in \Xi$, so we get that
\begin{eqnarray}
\label{petitjoliutile}
E(V(x+ \xi 1_{|\xi|\leq \tilde{K}(x)} Y)\vert {\cal H}) \geq E(V(x +\xi Y)\vert {\cal H}) \mbox{ a.s.}
\end{eqnarray}

\begin{proof}[of Lemma \ref{Kbound}.]
Fix some $ x \in \mathbb{R}$ and take $\xi\in\Xi$ such that $\xi \in D$ a.s. and $\vert\xi\vert\geq \max(1,x^+)$.
By \eqref{bajka}, we have
the following estimation:
\begin{eqnarray*}
V(x+ \xi Y) & = & V(x+\xi Y)1_{\{V(x+\xi Y) \geq 0\}}+V(x+\xi Y)1_{\{V(x+\xi Y)< 0\}} \\
 & \leq &  1_{\{V(x+\xi Y)\geq 0\}}\left(\vert\xi\vert^{\overline{\gamma}}V\left(\frac{x^+}{\vert\xi\vert}+\frac{\xi}
{\vert\xi\vert}Y\right)
+C\vert\xi\vert^{\overline{\gamma}}\right) +V(x+\xi Y)1_{\{V(x+\xi Y)< 0\}} \mbox{ a.s.}
\end{eqnarray*}
We start with the estimation using the positive part of $V$.
The random variable $L$ (recall \eqref{duracell}) is finite by \eqref{nemvegtelen}.
Thus, as $V$ is nondecreasing (see Assumption \ref{hypV1}), we obtain that a.s.
\begin{eqnarray*}
E\left(1_{\{V(x+\xi Y) \geq 0\}}V\left(\frac{x^+}{\vert\xi\vert}+\frac{\xi}
{\vert\xi\vert}Y\right)\vert {\cal H}\right)\leq E\left(V^+\left(1+\frac{\xi}
{\vert\xi\vert}Y\right)\vert {\cal H}\right) \leq L.
\end{eqnarray*}
For the estimation of the negative part, we introduce the event
\begin{eqnarray}
\label{Bset}
B:=\left\{ V(x+\xi Y) <0,\ \frac{\xi}{\vert\xi\vert} Y < -\alpha,\ V(-N)<-\frac{2C}{\beta}-1\right\}.
\end{eqnarray}
Then,  using \eqref{bajka-}, we obtain that a.s.
\begin{eqnarray*}
 -V(x+\xi Y)1_{\{V(x+\xi Y)< 0\}} & \geq &  -V(x+\xi Y)1_{B}\\
  & \geq &  -1_{B} \left(\vert\xi\vert^{\eta \underline{\gamma}} V\left(\frac{x^+}{\vert\xi\vert^{\eta}}+ \frac{\xi}
{\vert\xi\vert} Y\vert\xi\vert^{1-\eta}\right)+C\vert\xi\vert^{\eta \underline{\gamma}}\right).
\end{eqnarray*}
Now, from Assumption \ref{hypV3}, for all $\xi \in {\Xi}$ such that $\xi \in D$ a.s., we have (recalling Assumption \ref{AOAone}), a.s.:
\begin{eqnarray}
\nonumber
P\left(\left\{\frac{\xi}{\vert\xi\vert} Y < -\alpha,\ V(-N)<-\frac{2C}{\beta}-1\right\} \vert\mathcal{H}\right) &
\geq & P\left(V(-N)<-\frac{2C}{\beta}-1 \vert\mathcal{H}\right)
\\
\nonumber
&+& P(V(\xi Y< -\alpha \vert\xi\vert \vert\mathcal{H}) -1 \\
\nonumber
& \geq & 1-\beta/2+\beta -1 \\
\label{bantu}
 & \geq & \beta/2.
\end{eqnarray}
It is clear that $B$ contains
$$
\left\{ x^+-\alpha \vert\xi\vert<-N,\ \frac{\xi}{\vert\xi\vert} Y < -\alpha,\ V(-N)<-\frac{2C}{\beta}-1\right\}.
$$
Thus if we assume that $x^+-\alpha \vert\xi\vert\leq-N$, we get that $P(B|\mathcal{H})\geq \beta/2$ a.s.
Now assume that both  $x^+-\alpha \vert\xi\vert\leq-N$ and $\frac{x^+}{\vert\xi\vert^{\eta}}-\vert\xi\vert^{1-\eta}
\alpha\leq -N$ hold.
This is true if
$|\xi|\geq K_0 (x):= \max(1,x^+, \left(\frac{x^+ +N}{\alpha}\right)^{\frac1{1-\eta}}, \frac{x^+ +N}{\alpha})$
(recall that $0<\eta<1$ and we have assumed $\vert\xi\vert\geq \max(1,x^+)$).
Then
we have that a.s.,
\begin{eqnarray*}
E(V(x+\xi Y )1_{\{V(x+\xi Y)< 0\}}\vert\mathcal{H}) & \leq &
\vert\xi\vert^{\eta \underline{\gamma}}E(1_B V(-N)\vert\mathcal{H})+  C\vert\xi\vert^{\eta \underline{\gamma}}E(1_B\vert \mathcal{H})\\
 & \leq & -(\beta/2) \vert\xi\vert^{\eta \underline{\gamma}}.
\end{eqnarray*}
Putting together our estimations, for $\vert\xi\vert\geq K_0 (x)$ we have a.s.
\begin{eqnarray*}
E(V(x+\xi Y)\vert\mathcal{H}) & \leq &
\vert\xi\vert^{\overline{\gamma}}L +
C\vert\xi\vert^{\overline{\gamma}} -\frac{\beta}2 \vert\xi\vert^{\eta \underline{\gamma}}.
\end{eqnarray*}
In order to get \eqref{petitjoli}, it is enough to have, a.s.,
\begin{eqnarray}
\nonumber \vert\xi\vert^{\overline{\gamma}}L-\frac{\beta}6  \vert\xi\vert^{\eta \underline{\gamma}} & < & 0\\
\nonumber C\vert\xi\vert^{\overline{\gamma}} -\frac{\beta}6 \vert\xi\vert^{\eta \underline{\gamma}} & < & 0\\
\label{glasgow} -\frac{\beta}6 \vert\xi\vert^{\eta \underline{\gamma}} -E(V(-x^-)|\mathcal{H})& < & 0.
\end{eqnarray}
Since $\overline{\gamma}<\eta \underline{\gamma}<\underline{\gamma}$, the first two inequalities will be satisfied as
soon as $|\xi|\geq K_1 (x)$ (recall \eqref{K1}) and the last one as soon as
$|\xi|\geq K_2 (x)$ (recall \eqref{K2}).
From Assumption \ref{AOAone}, $\alpha$ and $\beta$  are $\mathcal{H}$-measurable random variables such that $\alpha>0$ and $\beta>0$ a.s. so $1/\alpha$ and $1/\beta$ are a.s. finite-valued. As $N$ and $L$ are also an $\mathcal{H}$-measurable and finite random variables,  so is  $K_1 (x)$.
It is also clear that $K_1(\omega,x)$ is non-decreasing in $x$.
Moreover, from Assumption \ref{hypV1}, $K_2(\omega,x)$ is non-increasing in $x$ and
from Assumption \ref{AOAone},
$K_2(\cdot,x)$ is clearly $\mathcal{H}$-measurable.
As $[E(V(-x^-)|\mathcal{H})]^- \leq E(V^-(-x^-)|\mathcal{H})$, by \eqref{veges} $K_2(x)$ is a.s. finite valued.

Let $\hat{K} (x) =  \max(K_1 (x),K_2 (x))$. Then
\eqref{petitjoli}
is satisfied if $\vert\xi\vert\geq \hat{K} (x)$. From the monotonicity property of $K_1(\omega,\cdot)$ and $K_2(\omega,\cdot)$, we get that
$\tilde{K}(x)\geq \hat{K}(x)$. Thus \eqref{petitjoli}
is also satisfied as soon as $\vert\xi\vert\geq \tilde{K} (x)$.

The random function $\tilde{K}(\cdot)$
is trivially $\mathcal{H}\otimes\mathcal{B}(\mathbb{R})$-measurable (and a.s. constant on intervals of the form $[n,n+1)$, $n\in\mathbb{Z}$).

By \eqref{duracell}-\eqref{Ktilde}, $\tilde{K}(x)$ is dominated by a polynomial function of
$(\lfloor x\rfloor+1)^+,N,1/\alpha,1/\beta,L$ and $[E(V(-\lfloor x\rfloor^-)\vert\mathcal{H})]^-$.
When $V(x)\geq -m(1+|x|^p)$,
$[E(V(-\lfloor x\rfloor^-)\vert\mathcal{H})]^- \leq  m(|\lfloor x\rfloor|^p+1)$ a.s.
So $\tilde{K}(x)$ is a.s. dominated by
a polynomial function in $|\lfloor x\rfloor|$, i.e. $\tilde{K}(x)\leq M'  (|\lfloor x\rfloor|^{\theta}+1)$ a.s.
for some
$\theta>0$ and for some random variable $M'$ which is a polynomial function of
$N,1/\alpha,1/\beta$ and $L$.  Thus $M'$ is a non-negative,  a.s. finite valued and
$\mathcal{H}$-measurable random variable.

As $\mathbb{R}=\cup_{n \in \mathbb{Z}}[n,n+1)$ and for all
$x \in [n,n+1)$, $\tilde{K}(x)\leq M'  (|n|^{\theta}+1)$ a.s.  one can find a common full measure set
on which $\tilde{K}(x)\leq M'  (|\lfloor x\rfloor |^{\theta}+1) \leq  M  (|x|^{\theta}+1)$ where
$M =(2^{\theta} +1) M'$ from the
simple estimation
$|\lfloor x\rfloor |^{\theta}\leq ||x| + 1|^{\theta}\leq 2^{\theta}(|x|^{\theta} + 1)$.
\end{proof}

\begin{remark}  A predecessor of Lemma \ref{Kbound} above is Lemma 4.8 of \cite{RS05} whose arguments, however,
are considerably simpler since $V$ is assumed concave in \cite{RS05}. We indicate here a correction
to that Lemma: in the estimates one needs to change the term $2C|\xi|^{\gamma}$ (appearing twice) to
$C[|\xi|^{\gamma}+|\xi|^{\gamma(1+\gamma)/2}]$.
\end{remark}

\begin{lemma}\label{bound}
Let Assumptions \ref{hypD}, \ref{AOAone}, \ref{hypV1}, \ref{hypVmoinsinfty}, \ref{hypV2} and \ref{hypV3} hold. Fix $x_0, x_1 \in \mathbb{R}$ with $x_0<x_1$. Then the $\mathcal{H}$-measurable, a.s. finite valued random variable $K=K(\omega, x_0,x_1)> 0$ (recall \eqref{K}) is
such that for all $x_0 \leq x\leq x_1$ we have:
\begin{equation}
\label{supborne}
-\infty<v(x) = \mathrm{ess.}\sup_{\xi\in \Xi, |\xi|\leq K}E(V(x+\xi Y)\vert\mathcal{H})<\infty\mbox{ a.s.}
\end{equation}
For any $\mathcal{H}$-measurable, positive, a.s. finite valued random variable $I$
there exists an $\mathcal{H}$-measurable, a.s. finite valued random variable $N'>0$ such that
$v(-N')\leq -I$ a.s. More precisely ${N}'$  is a polynomial function of $\frac{1}{\beta}$, $N$, $I$ and
$E(V^+(\bar{K} |Y| )\vert\mathcal{H})$, where
\begin{equation}\label{BARK}
\bar{K}:=\max \left(1,\frac{N}{\alpha},\left(\frac{N}{\alpha}\right)^{\frac1{1-\eta}},
\left(\frac{8L}{\beta}\right)^{\frac1{\eta \underline{\gamma}-\overline{\gamma}}},
\left(\frac{8C}{\beta}\right)^{\frac1{\eta \underline{\gamma}-\overline{\gamma}}}\right).
\end{equation}
\end{lemma}

\begin{proof}
Fix some $x_0 \leq x\leq x_1$.
First note that,
\begin{eqnarray*}
v(\omega,x)=\mathrm{ess.}\sup_{\xi\in \Xi, \xi \in D}E(V(x+\xi Y)\vert\mathcal{H})(\omega)\mbox{ a.s.}
\end{eqnarray*}
by Remark \ref{haromcsillag}. So from now on we assume that $\xi\in D$. We may as well assume
$D\neq \{0\}$ a.s. since the statement of this Lemma is clear on the event $\{D=\{0\}\}$.

Then the equality in \eqref{supborne} follows immediately from  \eqref{petitjoliutile}.
We now show that $v$ is finite. Let $\xi\in \Xi$, $|\xi|\leq K$,
\begin{eqnarray*}
-E(V^-(-|x|-K |Y|)\vert\mathcal{H})\leq E(V(x+\xi Y)\vert\mathcal{H})\leq E(V^+(|x|+K |Y|)\vert\mathcal{H}) \;\mbox{a.s.}
\end{eqnarray*}
and we conclude by Assumption \ref{hypVmoinsinfty}.

Looking carefully at the estimations of Lemma \ref{Kbound}, if $x <0$ and
$|\xi|\geq \max(1,\left(\frac{ N}{\alpha}\right)^{\frac1{1-\eta}},\frac{ N}{\alpha})$, we have that
\begin{eqnarray}
\label{eqneg}
E(V(x+\xi Y)1_{\{V(x+\xi Y)\geq 0\}}\vert\mathcal{H})+\frac12E(V(x+\xi Y)1_{\{V(x+\xi Y)< 0\}}\vert\mathcal{H})\leq 0 \; \mbox{ a.s.}
\end{eqnarray}
provided that
$\vert\xi\vert^{\overline{\gamma}}L
+C\vert\xi\vert^{\overline{\gamma}} -\frac{\beta}4 \vert\xi\vert^{\eta \underline{\gamma}}
\leq 0$.
So \eqref{eqneg} holds true provided that $\vert\xi\vert^{\overline{\gamma}}L-\frac{\beta}8 \vert\xi\vert^{\eta \underline{\gamma}}
\leq 0$,
and $C\vert\xi\vert^{\overline{\gamma}}
-\frac{\beta}8 \vert\xi\vert^{\eta \underline{\gamma}}\leq 0$, i.e.
\begin{eqnarray*}
|\xi|\geq  \max \left(1,\frac{N}{\alpha},\left(\frac{N}{\alpha}\right)^{\frac1{1-\eta}},
\left(\frac{8L}{\beta}\right)^{\frac1{\eta \underline{\gamma}-\overline{\gamma}}},
\left(\frac{8C}{\beta}\right)^{\frac1{\eta \underline{\gamma}-\overline{\gamma}}}\right)=\bar{K}.
\end{eqnarray*}

Let $I$ be an $\mathcal{H}$-measurable positive a.s. finite valued random variable, it remains to show that there exists a positive, a.s. finite valued  and
$\mathcal{H}$-measurable random variable $N'$ satisfying $v(-N')\leq -I$ a.s.
From now on we work on the event $\{x\leq -N\}$. Then a.s.,
\begin{eqnarray*}
-E(V(x+\xi Y )1_{\{V(x+\xi Y)< 0\}}\vert\mathcal{H}) & \geq & -E\left(1_{\left\{\frac{\xi}{\vert\xi\vert} Y < -\alpha,\ V(-N)<-\frac{2C}{\beta}-1\right\} } V(x-\alpha \vert\xi\vert )\vert\mathcal{H}\right)\\
    & \geq & -E\left(1_{\left\{\frac{\xi}{\vert\xi\vert} Y < -\alpha,\ V(-N)<-\frac{2C}{\beta}-1\right\} }V(x)\vert\mathcal{H}\right)\\
  & \geq & \left(\left(1+\frac{2C}{\beta}\right)\left(\frac{x}{-N}\right)^{\underline{\gamma}}- C\left(\frac{x}{-N}\right)^{\underline{\gamma}}\right)
  \frac{\beta}2\\ &\geq& \frac{\beta}2\left(\frac{x}{-N}\right)^{\underline{\gamma}},
\end{eqnarray*}
where we have used Assumption \ref{hypV2} (see \eqref{bajka-}), \eqref{bantu} and the fact that $\beta\leq 1$.
Thus, if $|\xi| \leq \bar{K}$, we obtain that
\begin{eqnarray}
\label{lundi}
E(V(x+\xi Y )\vert\mathcal{H})\leq E(V^+(\bar{K}|Y| )\vert\mathcal{H}) -
\frac{\beta}2\left(\frac{x}{-N}\right)^{\underline{\gamma}} \mbox{ a.s.}
\end{eqnarray}
Recall the definition of $\bar{K}$ and \eqref{eqneg}:
if $|\xi| \geq \bar{K}$ then we get that
\begin{eqnarray}
\label{mardi}E(V(x+\xi Y )\vert\mathcal{H})\leq \frac12E(1_{V(x+\xi Y)< 0}V(x+\xi Y )\vert\mathcal{H}) \leq  -
\frac{\beta}4\left(\frac{x}{-N}\right)^{\underline{\gamma}}\mbox{ a.s.}
\end{eqnarray}
The right-hand sides of both \eqref{lundi} and \eqref{mardi} are smaller than $ -I$ if
\begin{eqnarray}
\label{mercredi}
 \left(\frac{x}{-N}\right)^{\underline{\gamma}} & \geq & \frac{4}{\beta}\left(I+
 E(V^+(\bar{K}|Y| )\vert\mathcal{H})\right)\mbox{ a.s.}
\end{eqnarray}
We may and will assume that $I\geq 1/4$ which implies $4I/\beta\geq 1$. So there exists an $\mathcal{H}$-measurable random variable
\begin{eqnarray}
\label{jeudi}
N':=N
\left(\frac{4}{\beta}\left(I+
 E(V^+(\bar{K} |Y| )\vert\mathcal{H})\right)\right)^{\frac1{\underline{\gamma}}}\geq N \mbox{ a.s.},
\end{eqnarray}
such that, as soon as $x \leq -N'$,
$E(V(x+\xi Y )\vert\mathcal{H})\leq -I$ a.s. and, taking the supremum over all $\xi$, $v(x)\leq -I$ a.s. holds.
From \eqref{jeudi},  one can see that  ${N}'$  is a polynomial function of $\frac{1}{\beta}$, $N$, $I$ and
$E(V^+(\bar{K}|Y| )\vert\mathcal{H})$. ${N}'$ is also a.s. finite valued since $I$, $N$ and $1/\beta$  are (recall Assumption
\ref{AOAone}) and
(\ref{nemvegtelen}) holds true.
\end{proof}

\begin{lemma}
\label{candycontinue}
Let Assumptions \ref{hypV1} and \ref{hypVmoinsinfty} hold. There exists a version
$G(\omega,x,y)$ of $E(V(x+yY) \vert {\cal H})(\omega)$ for $(\omega,x,y) \in \Omega \times
\mathbb{R} \times  \mathbb{R}^d$ such that

(i) for almost all $\omega \in \Omega$,
$(x,y) \in \mathbb{R}\times  \mathbb{R}^d \rightarrow G(\omega,x,y) \in \mathbb{R}$ is continuous and nondecreasing in $x$;

(ii) for all $(x,y) \in \mathbb{R}\times  \mathbb{R}^d$, the function
$ \omega\in \Omega \rightarrow G(\omega,x,y) \in \mathbb{R}$ is ${\cal H}$-measurable;

(iii) for each $x  \in \mathbb{R}$ and for each $\mathcal{H}$-measurable $\xi$, we have that $E(V(x+\xi Y)\vert\mathcal{H})$ exists,
it is finite and
\begin{eqnarray}\label{cakk}
G(\cdot,x,\xi)=E(V(x+\xi Y)\vert\mathcal{H}), \mbox{ a.s.}
\end{eqnarray}
\end{lemma}
\begin{remark}
\label{castva} Note that, in particular, $G$ is $\mathcal{H}\otimes\mathcal{B}(\mathbb{R})\otimes\mathcal{B}(\mathbb{R}^d)$-measurable,
by p. 70 of \cite{CV77}.
\end{remark}
\begin{proof}[of Lemma \ref{candycontinue}.]
For part (i) of Lemma \ref{candycontinue}, we proceed in three steps.
First, we define a version of $(q,r)\to E(V(q+rY) \vert {\cal H})(\omega)$ which
is uniformly continuous on any precompact set $\mathbb{Q}^{d+1}\cap[-N,N]^{d+1}$, outside a $P$-zero set.
Then, in the second step, we  extend this version by continuity to $\mathbb{R}^{d+1}$ and in the third step
we show that this extension is, in fact,
a version of $(x,y)\to E(V(x+yY) \vert {\cal H})$, for all $x,y$.

\textit{Step 1:} Let us fix a version $G(\omega,q,r)$ of $E(V(q+rY)
\vert {\cal H})$ for all $(q,r)\in\mathbb{Q}^{d+1}$. Fix
$N>0$.

For each $r\in [-N,N]^{d}\cap \mathbb{Q}^{d}$ and
$q_1,q_2\in\mathbb{Q}\cap [-N,N]$ with $q_1\leq q_2$ we have $G(\omega,q_1,r)\leq G(\omega,q_2,r)$ a.s. by Assumption \ref{hypV1},
hence we can fix a set $\Omega'\subset\Omega$ of full measure such that $G(\omega,\cdot,r)$ is
nondecreasing over $\mathbb{Q}\cap [-N,N]$ for all $r\in [-N,N]^{d}\cap \mathbb{Q}^{d}$ and for all $\omega\in\Omega'$.

We claim that, for almost every $\omega$, the function
$(q,r)\to G(\omega,q,r)$ is uniformly continuous on
$[-N,N]^{d+1}\cap \mathbb{Q}^{d+1}$, i.e.,
\begin{eqnarray}
\label{thyro}
P\left(\cap_{\ell\in\mathbb{N}} M_{\ell} \right)=1,
\end{eqnarray}
where
\[
M_{\ell}:=\bigcup_{k\in\mathbb{N}}\bigcap_{(q_1,r_1),(q_2,r_2)\in [-N,N]^{d+1}\cap \mathbb{Q}^{d+1}, |q_1-q_2|+|r_1-r_2|\leq 1/k} \left\{|G(q_1,r_1)-G(q_2,r_2)|
\leq \frac1{\ell}\right\}.
\]
Fix $\ell \in\mathbb{N}$. By Assumption \ref{hypV1}, there exists a
full measure set $\Omega''$ such that $(x,y)\to
V(x+yY)$ is continuous and hence  uniformly continuous on
$[-N,N]^{d+1}$ for $\omega\in\Omega''$.  Define the events
$$
A_m(\ell)=
\bigcap_{(x,y),(z,w)\in [-N,N]^{d+1}\cap \mathbb{Q}^{d+1},\\ |x-z|+|y-w|<1/m}
\left\{\omega\in{\Omega}: |V(x+yY)(\omega)-V(z+wY)(\omega)|<\frac{1}{2\ell})\right\}.
$$
Uniform continuity implies that $\cup_m A_m(\ell)\supset{\Omega}''$. Define the disjoint sets
$$
B_1(\ell)=A_1(\ell),\ B_{m+1}(\ell)=A_{m+1}(\ell)\setminus\cup_{j=1}^m A_j(\ell)
$$
and set
$$
\zeta_{\ell}:=\sum_{m=1}^{\infty} \frac{1}{m} 1_{B_m(\ell)}.
$$
By construction, $\zeta_{\ell}$ is a random variable such
that on ${\Omega}''$,
\begin{eqnarray}
\label{lyro}
|V(x+yY)(\omega)-V(z+wY)(\omega)|\leq \frac{1}{2\ell}
\end{eqnarray}
whenever $(x,y),(z,w)\in
[-N,N]^{d+1}\cap\mathbb{Q}^{d+1}$ and $|x-z|+|y-w|\leq \zeta_{\ell}(\omega)$.

Now define
$$
\chi:=\sup_{(q,r)\in \mathbb{Q}^{d+1}\cap [-N,N]^{d+1}}|V(q+rY)|.
$$

As from Assumption \ref{hypV1}, $|V(q+rY)|= V^-(q+rY) + V^+(q+rY) \leq -V^-(-N-N|Y|) + V^+(N+N|Y|)$,  from Assumption  \ref{hypVmoinsinfty}, we get that:
\begin{equation}\label{firenze}
E\left(\chi\vert\mathcal{H}\right)<\infty
\end{equation}
holds almost surely. Hence, by Lemma \ref{lebbencs} (the conditional Lebesgue theorem), $E(\chi 1_{\{\chi\geq m\}}\vert\mathcal{H})\to 0$ as $m\to\infty$.
Fix versions $X_m$ of $E(\chi 1_{\{\chi\geq m\}}\vert\mathcal{H})$ and let $\Omega'''$ be the (full measure) set
where the above convergence holds.
The events
$$
C_m(\ell):=\left\{\omega\in\Omega''': X_m\leq \frac1{8\ell}\right\},\ m\in\mathbb{N},
$$
cover $\Omega'''$, satisfy $C_m(\ell) \in \mathcal{H}$ and we may define
$$
D_1(\ell)=C_1(\ell),\ D_{m+1}(\ell)=C_{m+1}(\ell)\setminus\cup_{j=1}^m C_j(\ell).
$$
Now set
$$
\eta_{\ell}:=\sum_{m=1}^{\infty} \frac{1}{8\ell m}1_{D_m(\ell)}.
$$
Note that, by construction,
\begin{equation}\label{bobbo}
E\left(\chi 1_{\{\chi  \geq \frac{1}{8\ell \eta_l}\}}\vert\mathcal{H}\right)\leq   \frac{1}{8\ell}\mbox{ a.s.}
\end{equation}

By a similar argument, we can choose an
$\mathcal{H}$-measurable $\mathbb{N}\setminus \{0\}$-valued random variable
$\psi_{\ell}$ such that
\begin{equation}\label{bobo}
P(1/\psi_{\ell}\geq
\zeta_{\ell}\vert\mathcal{H})\leq \eta_{\ell}\mbox{ a.s.}
\end{equation}
Define $A:=\{1/\psi_{\ell}\geq
\zeta_{\ell}\}$.
$\eta_{\ell}$ is clearly $\mathcal{H}$-measurable and
one has, almost surely,
\begin{eqnarray}\nonumber
E\left( 1_A\sup_{(q,r)\in \mathbb{Q}^{d+1}\cap [-N,N]^{d+1}}|V(q+rY)|\vert\mathcal{H}\right) & = &
E\left(\chi 1_{A \cap \{\chi \geq \frac{1}{8\ell \eta_{\ell}}\}}\vert\mathcal{H}\right)+
E\left(\chi 1_{A \cap \{\chi < \frac{1}{8\ell \eta_{\ell}}\}}\vert\mathcal{H}\right)\\
& \leq  & \frac{1}{8\ell} + \frac{1}{8\ell  \eta_{\ell}}P(A\vert\mathcal{H})
\leq \frac{1}{4\ell}.\label{bobbbo}
\end{eqnarray}

Let $\bar{\Omega}$ denote a full measure set where \eqref{bobbo}, \eqref{bobo}, \eqref{bobbbo} all hold. Define the sets $B=B(q_1,q_2,r_1,r_2,\ell):=\{\omega:|q_1-q_2|+|r_1-r_2|\leq 1/\psi_{\ell}(\omega)\}$.
By \eqref{lyro}, the definitions of $\eta_{\ell}$, $\psi_{\ell}$ and the above a.s. inequalities, we
have on a set $\Omega_{q_1,q_2,r_1,r_2}\subset \bar{\Omega}$ of full measure that
\begin{eqnarray*}\nonumber
1_B|G(\omega,q_1,r_1)-G(\omega,q_2,r_2)| & \leq &  E(1_B|V(q_1+r_1Y)-V(q_2+r_2Y)|\vert\mathcal{H})\\
 & \leq  &  1_B E\left(\frac{1}{2\ell}1_{
 \{
 \frac{1}{\psi_{\ell}}\leq \zeta_{\ell}
 \}
 }\vert\mathcal{H}\right)  +  \nonumber\\
& &
2\times 1_BE\left(1_{
\{\frac{1}{\psi_{\ell}}>\zeta_{\ell}
\}
} \sup_{(q,r)\in \mathbb{Q}^{d+1}
\cap [-N,N]^{d+1}}|V(q+rY)| \vert\mathcal{H}\right)\nonumber \\
& \leq& 1_B \left(\frac{1}{2\ell}+2\frac{1}{4\ell}\right)=1_B\frac{1}{\ell}.
\end{eqnarray*}
This shows that $B(q_1,q_2,r_1,r_2,\ell)\cap\Omega_{q_1,q_2,r_1,r_2}\subset \{|G(q_1,r_1)-G(q_2,r_2)|
\leq \frac1{\ell}\}$. Hence
\begin{equation}\label{bakk}
\left(\bigcup_{k\in\mathbb{N}}\bigcap_{(q_1,r_1),(q_2,r_2)\in [-N,N]^{d+1}\cap \mathbb{Q}^{d+1}, |q_1-q_2|+|r_1-r_2|\leq 1/k}
(B(q_1,q_2,r_1,r_2,\ell)\cap\Omega_{q_1,q_2,r_1,r_2})\right)
\end{equation}
is a subset of $M_{\ell}$.
Let $\omega \in {\Omega}$ arbitrary. Then for $k:=\psi_{\ell}(\omega)$,
$\omega \in B(q_1,q_2,r_1,r_2,\ell)$ for all $q_1,q_2,r_1,r_2$ such that
$|q_1-q_2|+|r_1-r_2|\leq 1/k$. In other words,
$$\Omega=\bigcup_{k\in\mathbb{N}}\bigcap_{|q_1-q_2|+|r_1-r_2|\leq 1/k} B(q_1,q_2,r_1,r_2,\ell)$$
and hence $M_{\ell}$  has full measure by \eqref{bakk} and \eqref{thyro} is proved.
Let $\tilde{\Omega}:=\Omega'\cap \left(\bigcap_{\ell} M_{\ell}\right)$. One gets that
for all $\omega \in \tilde{\Omega}$, the function
$(q,r)\to G(\omega,q,r)$ is uniformly continuous on
$[-N,N]^{d+1}\cap \mathbb{Q}^{d+1}$ and has the claimed monotonicity property as well. Note that $\tilde{\Omega}$ is a set of probability $1$.
 This concludes step 1.

\textit{Step 2:} Clearly, on $\tilde{\Omega}$, there is a unique extension by continuity of $G(\omega,x,y)$ over $[-N,N]^{d+1}$. Thus $G(\omega,x,y)$ can be
defined for
all $(x,y)\in\mathbb{R}^{d+1}$ in a continuous way on some $\tilde{\Omega}$ of full measure.
Note that, on $\tilde{\Omega}$, for all $q_1,q_2\in\mathbb{Q}$, $y\in\mathbb{Q}^d$, we have that
$$
G(\omega,q_1,y)\leq G(\omega,q_2,y)
$$
and this extends to $q_1,q_2\in\mathbb{R}$, $y\in\mathbb{R}^d$
by continuity.

\textit{Step 3:} It remains to show that, for all $(x,y)\in\mathbb{R}^{d+1}$,
$G(\omega,x,y)$ is a version of $E(V(x+yY) \vert {\cal H})(\omega)$.
To see this, let
$(q_n,r_n) \in \mathbb{Q}^{d+1}$ and $(x,y) \in \mathbb{R}^{d+1}$ be such that
$(q_n,r_n)$ tends to $(x,y)$. By continuity of $G$ on $\tilde{\Omega}$, $G(\omega,q_n,r_n)$ tends to $G(\omega,x,y)$ a.s.
By Assumption \ref{hypV1}, $V$ is almost surely continuous. So on a full measure set,
$V(q_n+r_nY)$ goes to $V(x+y Y)$.
Moreover, there exists some $n_0$ such that
for $n \geq n_0$, $x-1\leq q_n \leq x+1$ and $|r_n|\leq |y|+1$. As by Assumption \ref{hypV1}, $V$
is a.s. non-decreasing, we get that, on another full measure set,
\begin{eqnarray*}\label{bababa}
- V^-(x-1 - (|y|+1) |Y|)  & \leq V(q_n+r_nY) \leq &  V^+(x+1 + (|y|+1) |Y|).
\end{eqnarray*}
By Assumption \ref{hypVmoinsinfty}, we can apply Lemma \ref{lebbencs} (the conditional
Lebesgue theorem) and conclude that
$G(\omega,q_n,r_n)$ tends a.s. to $E(V(x+yY) \vert {\cal H})$: $G(\cdot,x,y)$ is a version of $E(V(x+yY) \vert {\cal H})$ and  \eqref{cakk} is proved for constants.

\textit{Step 4:} Assertion (ii) is straightforward, by
the definition of conditional expectations.

\textit{Step 5:} As for Assertion (iii), \eqref{cakk} is clear for
constants $\xi=x$ by step 3 above. We prove \eqref{cakk} for $\mathcal{H}$-measurable step functions
$\varsigma=\sum_ny_n1_{\varsigma=y_n}$ next.
It is clear that $1_{\varsigma=y_n}G(\omega,x,\varsigma)=
1_{\varsigma=y_n}G(\omega,x,y_n)=E(1_{\varsigma=y_n}V(x+y_n Y)\vert\mathcal{H})=
E(1_{\varsigma=y_n}V(x+\varsigma Y)\vert\mathcal{H})$ a.s.
So if we can apply Corollary \ref{nagyy} to $W=G(\omega,x,\varsigma)$, $Z=V(x+\varsigma Y)$ and
$A_n=\{\varsigma=y_n\}$, we can conclude that $G(\omega,x,\varsigma)=
E(V(x+\varsigma Y)\vert\mathcal{H})$ a.s.. This Corollary does apply since $E(1_{A_n}V(x+y_n Y)|\mathcal{H})$
exists a.s. and it is a.s. finite by Assumption \ref{hypVmoinsinfty}.

Now every $\mathcal{H}$-measurable random
variable $\xi$ can be approximated by a sequence of $\mathcal{H}$-measurable step functions
$(\varsigma_n)_n$ and we can conclude using (i) and Lemma \ref{lebbencs} as before.
\end{proof}

\begin{remark} An alternative way for constructing a suitable $G$ is using the theory of
 conditional expectations for normal integrands, see e.g. \cite{t} or \cite{chs}.
\end{remark}

\begin{lemma}
\label{candymaxiessup}
Let  Assumptions  \ref{hypD}, \ref{AOAone}, \ref{hypV1},  \ref{hypVmoinsinfty}, \ref{hypV2} and \ref{hypV3} hold. \\
Define $A(\omega,x)=\sup_{y \in \mathbb{Q}^d}G(\omega,x,y)$
for $(\omega,x) \in {\Omega} \times \mathbb{R}$. Let $A^{\tilde{K}}(\omega,x):=\sup_{y \in \mathbb{Q}^d, |y| \leq \tilde{K}(\omega,x) }G(\omega,x,y)$, where $\tilde{K}(\omega,x)$
is defined in \eqref{Ktilde}. Then we get that, on a set of full measure,

(i) the function $x \rightarrow  A^{\tilde{K}}(\omega,x)$, $x\in \mathbb{R}$ is non-decreasing and continuous,

(ii) $A^{\tilde{K}}(\omega,x)=A(\omega,x)$  for all $x\in \mathbb{R}$.

\noindent Finally, for each $x \in \mathbb{R}$,
\begin{equation}\label{marm}
v(x)=A(x)\mbox{ a.s.}
\end{equation}


\end{lemma}

\begin{remark}\label{vienna}
By \eqref{marm}, for each $x$, $A(x)$ is a version of $v(x)$ and hence, from this point on we may choose this version replacing $v(\cdot)$ by  $A(\cdot)$: by (i) and (ii), we will work with  a non-decreasing and continuous
 version of $v$.
\end{remark}

\begin{proof}[of Lemma \ref{candymaxiessup}.]
Fix some $\ell \in \mathbb{Z}$.
For $\ell \leq x < \ell+1$ and $\omega \in {\Omega}$, let $K_{\ell}=K(\omega,\ell,\ell+1)$ where $K(\omega,\ell,\ell+1)$
is defined in \eqref{K}. Let
$A^{K_{\ell}}(\omega,x):=\sup_{y \in \mathbb{Q}^d, |y| \leq K_{\ell} }G(\omega,x,y)$. We will first prove that, on a set of full measure,

(a) the function $x \rightarrow  A^{K_{\ell}}(\omega,x)$, $x\in [\ell,\ell+1)$ is non-decreasing and continuous,

(b) $A^{K_{\ell}}(\omega,x)=A(\omega,x)$  for all $x \in [\ell,\ell+1)$.

\noindent We prove (a) in two steps. First, we show that $x\to A^{K_{\ell}}(\omega,x)$  is continuous. Then we prove that
$q\to A^{K_{\ell}}(\omega,q)$ is non-decreasing on $\mathbb{Q}\cap [\ell,\ell+1)$.
By step 1, the monotonicity argument extends by continuity to $[\ell,\ell+1)$ and (a) is proved.
Note that we will work on the full-measure set $\tilde{\Omega}$ where all the conclusions  of Lemma \ref{candycontinue} {$(i)$} hold. Then we will prove \eqref{marm} and (b).

Now as $\mathbb{R}=\cup_{\ell \in \mathbb{Z}}[\ell,\ell+1)$ and by Lemma \ref{Kbound} $\tilde{K}|_{[\ell,\ell+1)}=K_{\ell}$, we see that
$A^{\tilde{K}}|_{[\ell,\ell+1)}=A^{K_{\ell}}.$ Thus it is still possible to find a full measure set such that (a) and (b) hold true on $\mathbb{R}$, (i) and (ii) and thus the lemma are proved.

Before all else we remark that $A^{K_{\ell}}(\omega,x),A(\omega,x)$ are $\mathcal{H}\otimes\mathcal{B}(\mathbb{R})$-measurable.
Indeed,  $A$ is ${\cal H}\otimes {\cal B}(\mathbb{R})$-measurable since $A$ is defined as a
countable supremum and by Remark \ref{castva} $G$ is ${\cal H}\otimes {\cal B}(\mathbb{R})\otimes {\cal B}(\mathbb{R}^d)$-measurable. Now one has
$$
A^{K_{\ell}}(\omega,x)=\sup_{y\in\mathbb{Q}^d} [1_R G(\omega,x,y)+1_{R^C} G(\omega,x,0)],
$$
where
$$
R:=\{(\omega,y): |y|\leq K_{\ell}\}.
$$
Hence it suffices to show $R\in \mathcal{H}\otimes\mathcal{B}(\mathbb{R}^d)$. As $\infty>K_{\ell}\geq 0$ a.s. and $K_{\ell}$
is  ${\cal H}$-measurable (see Lemma \ref{Kbound}),  there exists
a non-increasing sequence of step functions $K^n_{\ell}$ converging to
$K_{\ell}$. Let $K^n_{\ell}=\sum_{j=1}^{\infty} c_j 1_{A_j}$ where $A_j \in {\cal H}$. Clearly,
$$
R_n:=\{(\omega,y): |y|\leq K^n_{\ell}\}=\left(\cup_{j=1}^{\infty} (A_j\times \{y:|y|\leq c_j\})\right) \cup \left(\cap_{j=1}^{\infty} A^c_j\times \{0\}\right) \in \mathcal{H}\otimes\mathcal{B}(\mathbb{R}^d),
$$
and $R=\cap_n R_n$, showing what was claimed.

\textit{Step 1:} Fix some $x \in \mathbb{R}$ such that $\ell \leq x < \ell + 1$ and $\omega \in \tilde{\Omega}$.
Let $x_n\in [\ell, \ell+1)$ be a sequence of real numbers converging to $x$.
By definition of $A^{K_{\ell}}$, for all $k$, there exists some
$y_k(\omega,x) \in \mathbb{Q}^d$, $|y_k(\omega,x)| \leq K_{\ell}(\omega)$ and $G(\omega,x,y_k(\omega,x))
\geq A^{K_{\ell}}(\omega,x)-1/k$. Moreover,
one has that $A^{K_{\ell}}(\omega,x_n) \geq G(\omega,x_n,y_k(\omega,x))$ for all $n$, and by Lemma \ref{candycontinue} (i),
$$\liminf_n A^{K_{\ell}}(\omega,x_n) \geq G(\omega,x,y_k(\omega,x)) \geq A^{K_{\ell}}(\omega,x)-1/k,$$
and letting $k$ go to infinity,
\begin{eqnarray}
\label{lulu}
\liminf_n A^{K_{\ell}}(\omega,x_n) \geq A^{K_{\ell}}(\omega,x).
\end{eqnarray}
Note that $A^{K_{\ell}}(\omega,x_n)$ is defined as a supremum over a precompact set.  Thus there exists
$y_n^* (\omega) \in \mathbb{R}^d$, $|y_n^*(\omega)| \leq K_{\ell}(\omega)$ and $A^{K_{\ell}}(\omega,x_n)=G(\omega,x_n,y_n^*(\omega))$.
By compactness, there exists some $y^*(\omega)$ such that
some subsequence $y_{n_k}^*(\omega)$ of $y_{n}^*(\omega)$ goes to $y^*(\omega)$, $k\to\infty$, and $\limsup_n A^{K_{\ell}}(\omega,x_n)=\lim_{k}A^{K_{\ell}}(\omega,x_{n_k})$.
By Lemma \ref{candycontinue} (i), one gets
$$\limsup_n A^{K_{\ell}}(\omega,x_n)=G(\omega,x,y^*(\omega)) \leq A^{K_{\ell}}(\omega,x).$$
Recalling \eqref{lulu}, this concludes the proof of continuity for $A^{K_{\ell}}$.

\textit{Step 2:} We argue $\omega$-wise again. Let $q_1\leq q_2$ with $q_1,q_2\in [\ell, \ell+1)$.
By definition of $A^{K_{\ell}}$,  there exists some
$y_n^1(\omega) \in \mathbb{Q}^d$ satisfying
$G(\omega,q_1,y_n^1(\omega))
\geq A^{K_{\ell}}(\omega,q_1)-1/n$.
Moreover,
one has that $A^{K_{\ell}}(\omega,q_2) \geq G(\omega,q_2,y_n^1(\omega))$.
So, as by Lemma \ref{candycontinue} (i), $G(\omega,q_2,y_n^1(\omega)) \geq G(\omega,q_1,y_n^1(\omega))$,
we get that $A^{K_{\ell}}(\omega,q_2) \geq A^{K_{\ell}}(\omega,q_1)-1/n$. We conclude, by
letting $n$ go to zero, that the inequality $A^{K_{\ell}}(\omega,q_1)\leq A^{K_{\ell}}(\omega,q_2)$ holds on $\tilde{\Omega}$
for any pairs $q_1\leq q_2$
of rational numbers. By continuity of $A^{K_{\ell}}$, we obtain that
the inequality $A^{K_{\ell}}(\omega,x)\leq A^{K_{\ell}}(\omega,y)$ holds on $\tilde{\Omega}$
for any pairs $x\leq y$
of real numbers between $\ell$ and $\ell+1$. This concludes the proof of (a).

\textit{Step 3:} We now turn to the second part of Lemma \ref{candymaxiessup}.
Applying Lemma \ref{gast} to $F(\omega,y)=G(\omega,\hat{x},y)$ (see  Lemma \ref{candycontinue} (i) and (ii)) and  $K= K_{\ell}$
for some $\ell \leq \hat{x} <\ell+1$
(recall that $K_{\ell}$ is ${\cal H} $-measurable), we obtain that, almost surely,
\begin{eqnarray*}
\sup_{y \in \mathbb{Q}^d, |y| \leq K_{\ell}(\omega) } G(\omega,\hat{x},y)=
\mathrm{ess.}\sup_{\xi\in \Xi, |\xi|\leq  K_{\ell}} G(\omega,\hat{x},\xi(\omega))
\end{eqnarray*}
Now applying the same Lemma \ref{gast} to $F(\omega,y)=G(\omega,\hat{x},y)$
for some $ \hat{x} \in \mathbb{R}$  and
$K=\infty$, we obtain that, almost surely,
\begin{eqnarray*}
\sup_{y \in \mathbb{Q}^d } G(\omega,\hat{x},y)=
\mathrm{ess.}\sup_{\xi\in \Xi} G(\omega,\hat{x},\xi(\omega)).
\end{eqnarray*}
Now from the definition of $v$, $A$ and \eqref{cakk} we obtain for each $\hat{x} \in \mathbb{R}$,
$$v(\hat{x})=\mathrm{ess.}\sup_{\xi\in \Xi}E(V(\hat{x}+\xi Y)\vert\mathcal{H})=
\mathrm{ess.}\sup_{\xi\in \Xi} G(\cdot,\hat{x},\xi)=A(\hat{x}) \;\mbox{ a.s.}$$
and \eqref{marm} is proved for all  $ \hat{x} \in \mathbb{R}$.
Using also Lemma \ref{bound}, \eqref{cakk} and the definition of $A^{K_{\ell}}$, we obtain for each $\ell \leq \hat{x} <\ell+1$,
$$v(\hat{x})=\mathrm{ess.}\sup_{\xi\in \Xi, |\xi|\leq K_{\ell}}E(V(\hat{x}+\xi Y)\vert\mathcal{H})=
\mathrm{ess.}\sup_{\xi\in \Xi, |\xi|\leq K_{\ell}} G(\cdot,\hat{x},\xi)=A^{K_{\ell}}(\hat{x})\;\mbox{ a.s.}$$

\textit{Step 4:}
Our considerations so far imply that
the set $\{A(\cdot,q)=A^{K_{\ell}}(\cdot,q)\mbox{
for all }q\in\mathbb{Q}\cap[\ell, \ell+1)\}$ has probability one. Fix some $\omega_0$ in the intersection of this set with
the one where  $A^{K_{\ell}}$ is non-decreasing and continuous (namely $\tilde{\Omega}$), this intersection is again a set of full measure.
For any $x\in[\ell, \ell+1)$, there exist some sequences $(q_n)_n, (r_n)_n \subset \mathbb{Q}$ such that
$q_n\nearrow x$ and $r_n \searrow x$.
As $A(\omega_0,\cdot)$ is non-decreasing {on $\mathbb{Q}$}
(by definition of $\omega_0$):
$$\lim_{q_n\nearrow x} A(\omega_0,q_n)=A(\omega_0,x-) \mbox{ and } \lim_{r_n \searrow x}
A(\omega_0,r_n)=A(\omega_0,x+).$$
As $A^{K_{\ell}}$ is continuous on $[\ell,\ell+1)$,
$$\lim_{q_n\nearrow x} A^{K_{\ell}}(\omega_0,q_n)= \lim_{r_n \searrow x}
A^{K_{\ell}}(\omega_0,q_n)=A^{K_{\ell}}(\omega_0,x).$$
So by choice of $\omega_0$, $A(\omega_0,x-)=A^{K_{\ell}}(\omega_0,x)=A(\omega_0,x+)$ hence
$\omega_0 \in \{A(\cdot,x)=A^{K_{\ell}}(\cdot,x)\mbox{
for all }x\in[\ell, \ell+1)\}$. Thus
$P(A(\cdot,x)=A^{K_{\ell}}(\cdot,x)\mbox{ for all }x\in[\ell, \ell+1))=1$ and (b) is proved.
\end{proof}

\begin{lemma}
\label{candymaxi}
Let Assumptions \ref{hypD}, \ref{AOAone}, \ref{hypV1},  \ref{hypVmoinsinfty}, \ref{hypV2} and \ref{hypV3} hold.
There is a set of full measure $\hat{\Omega}$ and an ${\cal H}\otimes {\cal B}(\mathbb{R})$-measurable sequence $\xi_n(\omega,x) $ such that
for all $\omega \in \hat{\Omega}$ and $x\in\mathbb{R}$,
\begin{eqnarray*}
\xi_n(\omega,x) & \in & D(\omega)\\
\vert\xi_n(\omega,x)\vert & \leq & \tilde{K}(\omega,x)\\
G(\omega,x,\xi_n(\omega,x)) & \to & A(\omega,x),
\end{eqnarray*}
see \eqref{Ktilde} for the definition of $\tilde{K}(\cdot)$.
Moreover, for
$(\omega,x) \in \hat{\Omega} \times \mathbb{R}$ define
\begin{eqnarray}
\label{servicepublique}
E_n(\omega,x):= |G(\omega,x,\xi_n(\omega,x)) - A(\omega,x)|.
\end{eqnarray}
Then $E_n$ is $\mathcal{H}\otimes\mathcal{B}(\mathbb{R})$-measurable.
For all $N>0$ and for all $\omega \in \hat{\Omega}$, $\sup_{|x| \leq N}E_n(\omega,x)\to 0$, $n\to\infty$.
\end{lemma}
\begin{proof} Choose $\tilde{\Omega}$ such that all the conclusions of Lemmata \ref{candycontinue} (i) and \ref{candymaxiessup}
hold on this set.

\textit{Step 1: construction of the sequence $(\xi_n)_n$}.\\
Let $q_1, \ldots, q_k, \ldots$ be an enumeration of $\mathbb{Q}^d$.
Define $\mathbb{D}_n:=\{l/2^n:l\in\mathbb{Z}\}$.

Recall from Assumption \ref{hypD} that, for almost all $\omega$, $D(\omega)$
is a non-empty vector subspace of $\mathbb{R}^d$ (and is thus closed).
For all $k$, consider the projection $Q_k(\omega)$ of $q_k$ on $D(\omega)$.
Then $Q_k \in D$ and, as in Proposition 4.6 of \cite{RS05}, the measurable selection theorem
(see for example Proposition III.44 in \cite{dm}) implies that the projection of any
${\cal H}$-measurable random variable on $D$ (a fortiori the projection of any constant) is ${\cal H}$-measurable.
Moreover from Remark \ref{haromcsillag}, $q_k Y=Q_k Y$ a.s. for all $k$. So we denote by $\hat{\Omega}$ the intersection of $\tilde{\Omega}$ with
$\cap_{k \in \mathbb{N}} \{q_k Y=Q_k Y\}$: it is again a set of full measure.

Let $C_1^n  =
\{(\omega,x) \in \hat{\Omega} \times \mathbb{D}_n \ : \ |q_1| \leq \tilde{K}(\omega,x) \mbox{ and } |G(\omega,x,q_1)- A(\omega,x)| <1/n \}$ and for
all $k \geq 2$, define $C_k^n$ recursively by
\begin{eqnarray*}
C_k^n & = & \{(\omega,x) \in \hat{\Omega} \times \mathbb{D}_n \ : \ |q_k| \leq \tilde{K}(\omega,x) \mbox{ and } |G(\omega,x,q_k)- A(\omega,x)| <1/n \}
\setminus \cup_{l=1,\ldots,k-1} C_l^n.
\end{eqnarray*}
As from Lemma \ref{Kbound} $\tilde{K}$ is ${\cal H}\otimes {\cal B}(\mathbb{R})$-measurable, $C_k^n$ is in
${\cal H}\otimes {\cal B}(\mathbb{R})$ (recall also Remark \ref{castva}). As from Lemma \ref{candymaxiessup},  $A(\omega,x)=A^{\tilde{K}}(\omega,x)=\sup_{q_k, |q_k| \leq \tilde{K}(\omega,x) }G(\omega,x,q_k)$,
one has $\cup_k C_k^n =\hat{\Omega} \times \mathbb{D}_n$. Define for $(\omega,x) \in \hat{\Omega} \times \mathbb{R}$
\begin{eqnarray}
\label{defxi}
\xi_n(\omega,x) & = & \sum_{k=1}^{\infty}\sum_{l=-\infty}^{\infty}Q_k (\omega)
1_{\{(\omega,l/2^n)\in C_k^n\}}(\omega)1_{\{l/2^n\leq x<(l+1)/2^n\}}(x).
\end{eqnarray}
Then $\xi_n$ is ${\cal H}\otimes {\cal B}(\mathbb{R})$-measurable. Fix some $n,l$ and $x\in [l/2^n,(l+1)/2^n)$. 
Then one has on $\{\omega \in \hat{\Omega}:(\omega,l/2^n)\in C^n_k\}$ (recall that $Q_k(\omega)$ is  the orthogonal projection of $q_k$ on $D(\omega)$),
$$|\xi_n(\omega,x)| =|Q_k(\omega)| \leq |q_k|\leq \tilde{K}(\omega,x).$$
Moreover, again on $\{\omega \in \hat{\Omega}:(\omega,l/2^n)\in C^n_k\}$, we get that by definition of $\hat{\Omega}$
\begin{eqnarray*}
G(\omega,x,\xi_n(\omega,x))& = & G(\omega,x,Q_k(\omega)) = E(V(x+Q_k(\omega)Y) \vert {\cal H})\\
   & = & E(V(x+q_k Y) \vert {\cal H})= G(\omega,x,q_k).
\end{eqnarray*}
As $\mathbb{D}_n$ is a countable set and the $C_k^n$ form a partition of
$\hat{\Omega}\times \mathbb{D}_n$, we thus have for all $n$ and $(\omega,x)\in\hat{\Omega} \times\mathbb{D}_n$
\begin{eqnarray*}
|\xi_n(\omega,x)| & \leq & \tilde{K}(\omega,x)\\
|G(\omega,x,\xi_n(\omega,x))- A(\omega,x)|  & < & 1/n.
\end{eqnarray*}

\textit{Step 2: proof of convergence.}\\
Fix any integer $N>0$, we will prove that for all $\omega\in\hat{\Omega}$, $\sup_{|x| \leq N}E_n(\omega,x)$ goes to
zero.
We argue for each fixed $\omega\in\hat{\Omega}$. As $A(\omega,x)$ is continuous
from Lemma \ref{candymaxiessup},
it is uniformly continuous on $[-N,N]$. The same argument applies to $G(\omega,x,y)$ on $[-N,N]\times
[-K(-N,N+1),K(-N,N+1)]^{d}$
(see Lemma \ref{candycontinue} (i) and the definition of $K(\cdot,\cdot)$ in
\eqref{K}). Hence for each $\epsilon>0$ there is $\eta(\omega)>0$
such that $|A(\omega,x)-A(\omega,x_0)|<\epsilon/3$ and
$|G(\omega,x,y)-G(\omega,x_0,y_0)|<\epsilon/3$
if $|x-x_0|+|y-y_0|<\eta(\omega)$. Now let $d_n(x)$
denote the element of $\mathbb{D}_n$ such that $d_n(x)\leq x<d_n(x)+(1/2^n)$. Then $\xi_n(\omega,d_n(x))=\xi_n(\omega,x)$.
Since $|\xi_n(\cdot,x)|\leq \tilde{K}(x)\leq K(-N,N+1)$ for all $x\in [-N,N]$,
we have
\begin{eqnarray*}
|G(\omega,x,\xi_n(\omega,x))- A(\omega,x)| & \leq  & |G(\omega,x,\xi_n(\omega,x))- G(\omega,d_n(x),
\xi_n(\omega,d_n(x))| + \\
& &  |G(\omega,d_n(x),
\xi_n(\omega,d_n(x))- A(\omega,d_n(x))|+\\
& & |A(\omega,d_n(x))- A(\omega,x)| \\
&\leq &
\epsilon/3 + 1/n + \epsilon/3 \leq \epsilon,
\end{eqnarray*}
if $n$ is chosen so large that both $1/2^n<\eta(\omega)$ and
$1/n<\epsilon/3$. To complete the proof it remains to show that $E_n$ is $\mathcal{H}\otimes\mathcal{B}(\mathbb{R})$-measurable.
Recalling Lemma \ref{candycontinue}, for almost all $\omega \in \Omega$,
$(x,y) \in \mathbb{R}\times  \mathbb{R}^d \rightarrow G(\omega,x,y) $ is continuous
and from Remark \ref{castva} $G$ is $\mathcal{H}\otimes\mathcal{B}(\mathbb{R})\otimes\mathcal{B}(\mathbb{R}^d)$-measurable.  As $\xi_n$ is ${\cal H}\otimes {\cal B}(\mathbb{R})$-measurable, $(\omega,x) \in \Omega \times \mathbb{R}\rightarrow G(\omega,x,\xi_n(\omega,x)) $ is ${\cal H}\otimes {\cal B}(\mathbb{R})$-measurable.
By definition ($A$ is a countable supremum of $\mathcal{H}\otimes\mathcal{B}(\mathbb{R})$-measurable functions), $A$ is also $\mathcal{H}\otimes\mathcal{B}(\mathbb{R})$-measurable,
and so is $E_n$.
\end{proof}

These preparations allow us to prove the existence of an optimal strategy:
\begin{proposition}
\label{candy}
Let Assumptions \ref{hypD}, \ref{AOAone}, \ref{hypV1},  \ref{hypVmoinsinfty}, \ref{hypV2} and \ref{hypV3} hold. Then there exists an
$\mathcal{H}\otimes\mathcal{B}(\mathbb{R})$-measurable
$\tilde{\xi}(\omega,x) \in D$ such that for each $x$,
\begin{eqnarray}
\label{supxi}
v(\omega,x) & = & E(V(x+\tilde{\xi}(\omega,x)Y) \vert\mathcal{H})\mbox{ a.s.}
\end{eqnarray}
Recall the definition of $\tilde{K}(x)$ from \eqref{Ktilde}. We have
\begin{eqnarray}
\label{macha}
|\tilde{\xi}(\omega,x)| & \leq & \tilde{K}(\omega,x)\mbox{ for all }x\in\mathbb{R}\mbox{ and }\omega\in\Omega.
\end{eqnarray}
The $\tilde{\xi}$ we have constructed satisfies
\begin{eqnarray}
\label{gateau}
A(\omega,H)=E(V(H+\tilde{\xi}(H)Y)\vert\mathcal{H})=
\mathrm{ess.}\sup_{\xi\in
\Xi}E(V(H+{\xi}Y)\vert\mathcal{H})\ \mathrm{a.s.},
\end{eqnarray}
for each $\mathcal{H}$-measurable $\mathbb{R}$-valued random variable $H$.
\end{proposition}
\begin{proof}
From Lemma \ref{candymaxi}, there exists a sequence $\xi_n(\omega,x) \in D$ such that
$G(\omega,x,\xi_n(\omega,x))$ converges to $A(\omega,x)$ for all $\omega \in \hat{\Omega}$ for some $\hat{\Omega}$ of full
measure and for all $x\in\mathbb{R}$.
Note that $|\xi_n(x)|$ is bounded by $\tilde{K}(x)$ for all $x \in \mathbb{R}$ and $\omega \in \hat{\Omega}$.

From Lemma A.2 of \cite{RS05} (see also Lemma 2 in \cite{KS01}), we find a random subsequence
$\tilde{\xi}_k(\omega,x)$
of $\xi_n(\omega,x)$ converging to some $\tilde{\xi}(\omega,x)$ for all $x$
and $\omega\in\Omega'$ for a set of  full measure  $\Omega'$ as $k\to\infty$. On the set $\Omega\setminus\Omega'$ we define
$\tilde{\xi}(\omega,x):=0$ for all $x$. Note that this ensures $|\tilde{\xi}(\omega,x)|\leq \tilde{K}(x)$ for all $x \in \mathbb{R}$ and $\omega \in {\Omega}$ and \eqref{macha} is proved.

Here $\tilde{\xi}_k(\omega,x)=\xi_{n_k}(\omega,x)=
\sum_{l \geq k}\xi_{l}(\omega,x)1_{\tilde{B}(l,k)}$,
with $\tilde{B}(l,k)=\{(\omega,x):\, n_k(\omega,x)=l \} \in {\cal H}\otimes\mathcal{B}(\mathbb{R})$ and
$\cup_{l \geq k}\tilde{B}(l,k)={\Omega}'\times\mathbb{R}$. Fix $x\in\mathbb{R}$ now.
Define $B(l,k):=\{\omega:(\omega,x)\in \tilde{B}(l,k)\}\in\mathcal{H}$.
Then we have that a.s.
\begin{eqnarray}
\label{myro}
E(V(x+\tilde{\xi}_k(x) Y)\vert\mathcal{H}) & = & \sum_{l \geq k} 1_{B(l,k)}E(V(x+\xi_{l}(x) Y)\vert\mathcal{H}) \\
\nonumber
 & \geq & \sum_{l \geq k} 1_{B(l,k)} (A(\omega,x) -E_l(\omega,x))\\
 \nonumber
  & \geq & \sum_{l \geq k} 1_{B(l,k)} (A(\omega,x) -\sup_{m\geq k}E_m(\omega,x))=A(\omega,x) -\sup_{m\geq k}E_m(\omega,x).
\end{eqnarray}
Here \eqref{myro} will be verified shortly, using Corollary \ref{nagyy}. The first inequality follows
from \eqref{cakk} and Lemma  \ref{candymaxi} (see \eqref{servicepublique}).

In \eqref{myro} we applied Corollary \ref{nagyy} with $W=\sum_{l \geq k} 1_{B(l,k)}
E(V(x+\xi_{l}(x) Y)\vert\mathcal{H})$, $A_l=B(l,k)$, $l\geq k$ and $Z=V(x+\tilde{\xi}_k(x) Y)$.
By Remark \ref{dunoon}, $E(Z1_{A_l}\vert\mathcal{H})$ exists and is a.s. finite.
Since $W1_{A_l}=E(Z1_{A_l}\vert\mathcal{H})$ a.s. holds true trivially, \eqref{myro} is satisfied.

Note that $E_m(\omega,x)\to 0$ a.s.,
$m\to\infty$  (see Lemma \ref{candymaxi}) also implies $\sup_{m\geq k}E_m(\omega,x)\to 0$ a.s., $k\to\infty$.
As $E(V(x+\tilde{\xi}_k(x) Y)\vert\mathcal{H}) \leq E(V^+(x+ \tilde{K}(x) |Y|) \vert\mathcal{H}) < \infty$ by \eqref{nemvegtelen}, the
(limsup) Fatou Lemma applies and we obtain, using Assumption \ref{hypV1}, that a.s.
\begin{eqnarray*}\nonumber
E(V(x+\tilde{\xi}(x) Y)\vert\mathcal{H})  & \geq  & \limsup_k E(V(x+\tilde{\xi}_k(x) Y)\vert\mathcal{H}) \\
 & \geq & \limsup_k (A(\omega,x) -\sup_{m \geq k} E_m(\omega,x))=A(\omega,x).\label{mancika}
\end{eqnarray*}
Recalling \eqref{marm}, \eqref{supxi} is proved for each $x$ since $v(x)\geq E(V(x+\tilde{\xi}(x) Y)\vert\mathcal{H})$ a.s. is trivial.

To see \eqref{gateau}, we will prove that the following inequalities hold true:
\begin{eqnarray}
\label{banine}
A(\omega,H)\leq E(V(H+\tilde{\xi}(H)Y)\vert\mathcal{H})\mbox{ a.s.}
\end{eqnarray}
and  for any fixed $\xi$
\begin{eqnarray}\label{montdor}
E(V(H+\xi Y)\vert\mathcal{H})\leq A(\omega,H)\ \mathrm{a.s.}
\end{eqnarray}
Then from \eqref{banine} and \eqref{montdor} applied to $\tilde{\xi}(H)$, we get that
$A(\omega,H)=E(V(H+\tilde{\xi}(H)Y)\vert\mathcal{H})$ a.s. Finally $A(\omega,H)=E(V(H+\tilde{\xi}(H)Y)\vert\mathcal{H})\leq
\mathrm{ess.}\sup_{\xi\in\Xi}E(V(H+\xi Y)\vert\mathcal{H}) \leq A(\omega,H)$ a.s.
(where the last inequality comes from \eqref{montdor} again) and \eqref{gateau} is proved.

\textit{Step 1: it is enough to prove \eqref{banine} for bounded $H$.}\\
As $H=\sum_{p=-\infty}^{\infty}H1_{p \leq H < p+1}$,  we want to apply Corollary \ref{nagyy} to
$W=A(\cdot,H)$, $A_p=\{p \leq H < p+1\}$ and $Z=V(H+\tilde{\xi}(H)Y)$ to conclude that
if \eqref{banine} is proved for each $H_p=H1_{p \leq H \leq p+1}$ then it is proved for $H$. We
only need to verify that $E(V(H+\tilde{\xi}(H)Y)1_{A_p}\vert\mathcal{H})$ exists and it is finite
a.s., but this is clear from Remark \ref{dunoon}.

\textit{Step 2: proof of \eqref{banine} for bounded $H$.}\\
First let us fix $p\in\mathbb{Z}$ such that $p\leq |H| < p+1$. Let us also fix $n$.
We will establish
that
\begin{equation}\label{momm}
A(\omega,H)-\tilde{E}_{n,p}(\omega)\leq E(V(H+{\xi}_n(H)Y)\vert\mathcal{H}) \mbox{ a.s.}
\end{equation}
where $\tilde{E}_{n,p}:=\sup_{p\leq x<p+1} E_n(\omega,x)$.
Recall that $E_n$ is defined in
\eqref{servicepublique} above and  is $\mathcal{H}\otimes\mathcal{B}(\mathbb{R})$-measurable. As the supremum may be taken over the
rationals, $\tilde{E}_{n,p}$ is $\mathcal{H}$-measurable.

As $H=\sum_{l=-\infty}^{\infty}H1_{\{H\in [l/2^n,(l+1)/2^n)\}}$, applying Corollary \ref{nagyy} again,
it is enough to prove \eqref{momm}
for $J^l=H1_{\{H\in [l/2^n,(l+1)/2^n)\}}$ for each $l=p2^n,\ldots,(p+1)2^n-1$.

Fix $l \in\{p2^n,\ldots,(p+1)2^n-1\}$. Fix some step functions $J^l_k=\sum_{m\geq 1}j_m^{k,l} 1_{J^{l}_k=j_m^{k,l}}$  converging to $J^l$, $k\to\infty$,
such that $j_m^{k,l}\in [l/2^n,(l+1)/2^n)$. Then, a.s.
\begin{eqnarray*}
E(V(j_m^{k,l}+\xi_{n}(j_m^{k,l}) Y)\vert\mathcal{H}) \geq  A(\omega,j_m^{k,l}) -\tilde{E}_{n,p}(\omega),
\end{eqnarray*}
from the construction of $\xi_n$ in Lemma \ref{candymaxi} (see \eqref{servicepublique}). So
 \eqref{momm}  holds for each $H=j_m^{k,l}$ and, applying Corollary \ref{nagyy},
\eqref{momm} holds also
for $H=J^l_k$.

From  \eqref{Ktilde} $\tilde{K}(x)=K(p,p+1)$ for $x\in [p,
p+1)$. By the construction of $\xi_n$ (see \eqref{defxi}), we have that ${\xi}_n(x)$ is constant for $x\in [l/2^n,
(l+1)/2^n)$ and thus ${\xi}_n(J^l_k)={\xi}_n(J^l)$. So
using the continuity of $A$ on the left-hand side, the continuity of $V$  and Fatou's lemma for the right-hand side, we get
that \eqref{momm}  holds for each $J^l$ and the statement \eqref{momm} is proved. Here we can use the limsup Fatou Lemma because
$V(J^l_k+{\xi}_n (J^l) Y) \leq   V^+(p+1+K(p,p+1) |Y|)$ and
the latter is $< \infty$ a.s. due to Assumption \eqref{nemvegtelen}.

Now we pass to the limit in \eqref{momm} along the random subsequence $n_k$
defined in the beginning of the proof  (again, \eqref{momm}  holds for $n_k$ by Corollary \ref{nagyy}).
From Lemma \ref{candymaxi}, $\tilde{E}_{n_k,p}\to 0$ a.s.
Recalling that, $\xi_{n_k}(\omega,x)$ converges to $\tilde{\xi}(\omega,x)$ for all $p \leq x < p+1$ on some
$\Omega'$ of full measure,
${\xi}_{n_k}(\omega,H(\omega))$ converges to $\tilde{\xi}(x, H(\omega))$ and
using the same Fatou-lemma argument, we get that
\eqref{banine} holds true with $H$ bounded.

\textit{Step 3: proof of \eqref{montdor}.}\\
Similarly as in step 1, it is enough to prove \eqref{montdor} for bounded $H$ and $\xi$.
We denote by $N$ the bound for $|\xi|$ and by $M$ the bound for $|H|$.
By construction of $A$ and \eqref{cakk}, \eqref{montdor} holds true for constant $H$, so by Corollary
\ref{nagyy} it holds true for step functions $H$.
Again, taking a sequence of step-function approximations $H_l\to H$
with $H_l$ uniformly bounded, using the continuity of
$A$ for the right-hand side and Fatou Lemma for the left-hand side (here it is liminf Fatou Lemma and we use that
$V(H_l+ {\xi}  Y) \geq   -V^-(-M- N |Y|)$ and
$E(V^-(-M -N |Y|)\vert\mathcal{H}) < \infty$ due to Assumptions \ref{hypV1} and \ref{hypVmoinsinfty}), we get
that \eqref{montdor}  holds for all bounded $H$, $\xi$ and hence for all $H$, $\xi$. The statement is proved.
\end{proof}

\begin{remark}
 For the proof of Theorem \ref{main} it would suffice to construct, for all $\mathcal{H}$-measurable $H$,
 some ${\xi}_H\in\Xi$ satisfying $E(V(H+{\xi}_H Y)|\mathcal{H})=A(H)$. We have obtained a much sharper result:
 there is $\tilde{\xi}:\Omega\times\mathbb{R}\to\mathbb{R}$ such that one can choose $\xi_H:=\tilde{\xi}(H)$ and this is what we use in Proposition \ref{camarche}.

 An alternative way for constructing $\xi_H$ is through the technology of normal integrands and measurable
 selection, as presented e.g. in Chapter 14 of \cite{rw}.
\end{remark}

\section{Dynamic programming}
\label{dyn}

We first prove that the  random functions associated to the dynamic programming procedure are well defined and finite
under appropriate integrability conditions.
\begin{proposition}\label{propre1}
Let $U:\mathbb{R}\to\mathbb{R}$ be non-decreasing and left-continuous.
Assume that  \eqref{negaU0} holds true.
Then the random functions $U_t$ (see \eqref{tuskes} and \eqref{vanek}) are well-defined recursively, for all $x\in\mathbb{R}$. Indeed, one can
choose $(-\infty,+\infty]$-valued versions which are a.s. non-decreasing and left-continuous (in $x$). In particular, each $U_t$
is $\mathcal{F}_{t}\otimes\mathcal{B}(\mathbb{R})$-measurable.
Moreover, for all $0\leq t \leq T$, almost surely for all $x \in \mathbb{R}$, we have:
\begin{eqnarray}
\label{utetu}
U_t(x)  \geq  U (x) >-\infty.
\end{eqnarray}
For all $1\leq t \leq T$, $x \in \mathbb{R}$, $\xi \in \Xi_{t-1}$, we obtain that a.s.
\begin{eqnarray}
\label{ut-x}
E(U_t^-(x+\xi\Delta S_{t})|\mathcal{F}_{t-1}) <  +\infty.
\end{eqnarray}
If we assume also that \eqref{lasagna} holds true then for all $1\leq t \leq T$ and $\xi \in \Xi_{t-1}$ we have for all $x$,
\begin{eqnarray}
\label{utx}
E(U_t(x+\xi\Delta S_{t})|\mathcal{F}_{t-1}) \leq U_{t-1}(x)<  +\infty\mbox{ a.s.}\\
\label{ut+x}
E(U_t^+(x+\xi\Delta S_{t})|\mathcal{F}_{t-1}) <  +\infty\mbox{ a.s.}
\end{eqnarray}
\end{proposition}
\begin{proof}
We prove the first part of the proposition under \eqref{negaU0} only.
At $t=T$, $U_T(x)  \geq  U (x)$ is by definition and \eqref{ut-x} holds true by \eqref{negaU0} and Lemma \ref{determ} applied with $V=U$, $Y=\Delta S_{t}$, $\mathcal{H}=\mathcal{F}_{t-1}$ and $H=x$.

Assume now that one can
choose an $(-\infty,+\infty]$-valued version of $U_{t+1}$ which is a.s. non-decreasing and left-continuous (in $x$). Assume also that the statements \eqref{utetu}, \eqref{ut-x} hold true at $t+1$.
Then Lemma \ref{propre}, applied with $V$ equal to this version of $U_{t+1}$, $Y=\Delta S_{t+1}$, $\mathcal{H}=\mathcal{F}_{t}$, provides an
increasing, left-continuous random function (namely $\mathfrak{A}(x)$ defined in Lemma \ref{propre}) which is a version of $U_t$. From
now on we work with this version of $U_t$.
Choosing $\xi=0$, we get that, for all $x \in \mathbb{R}$,
$$U_t(x)\geq
E(U_{t+1}(x)|\mathcal{F}_{t}) \geq U(x)>-\infty\mbox{ a.s.}$$
where the second inequality holds by the induction hypothesis \eqref{utetu}.
As both $U_t,U$ are left-continuous, $U_t(x)\geq U(x)$ holds for all $x$ simultaneously, outside
a fixed negligible set (see Lemma \ref{indis}).
This implies also that
$$E(U_t^-(x+\xi\Delta S_{t})|\mathcal{F}_{t-1}) \leq E(U^-(x+\xi\Delta S_{t})|\mathcal{F}_{t-1})<  +\infty,$$
by \eqref{negaU0} again. So $E(U_t(x+\xi\Delta S_{t})|\mathcal{F}_{t-1})$ is well-defined and statements
\eqref{utetu}, \eqref{ut-x} are proved for $U_t$.

Now we prove the second part of the proposition. For $x \in \mathbb{R}$ and for $0\leq j \leq T$, as
$U_j^-(x) \leq U^-(x) <\infty$ by \eqref{utetu} we get $E(U_j^- (x))<\infty$. Thus $E(U_j (x))$ is
well-defined and, by Lemma \ref{ujdonatuj},
$E(U_j (x)\vert\mathcal{F}_{j-1})$ is well-defined a.s., too, and
$$E(U_j (x))=E(E(U_j (x)\vert\mathcal{F}_{j-1}))$$
holds. Let
$\xi \in \Xi_{t-1}$, $1\leq t\leq T$. Choosing the strategy equal to zero at the dates $1,\ldots,t-1$, we get
\begin{eqnarray*}
\nonumber
E(U_0(x))\geq E(E(U_{1}(x)\vert\mathcal{F}_0))= E(U_1(x))\geq \ldots & \geq & E(E(U_{t-1}(x)\vert\mathcal{F}_{t-2})) \\
& = & E(U_{t-1}(x))
\geq   E(E(U_{t}
(x+\xi\Delta S_{t})\vert\mathcal{F}_{t-1})).
\end{eqnarray*}
As $E(U_0(x)) <\infty$, we obtain that $E(U_{t-1}(x))<\infty$, thus $U_{t-1}(x)<\infty$ a.s. and
\eqref{utx} as well as \eqref{ut+x} hold true.
\end{proof}

To perform a dynamic programming procedure, we need to establish that some
crucial properties of $U$ are true for $U_t$ as well, i.e. they are preserved by
dynamic programming. In particular the
``asymptotic elasticity''-type conditions \eqref{egyedik} and \eqref{kettedik}, see below.
\begin{proposition}\label{egeszR}
Assume that $U$ satisfies Assumption \ref{ae}. Then there is a
constant $C\geq 0$ such that for all
$x \in \mathbb{R}$ and $\lambda\geq 1$,
\begin{eqnarray}
\label{egyedik} U(\lambda x)& \leq &  \lambda^{\overline{\gamma}} U(x)+C\lambda^{\overline{\gamma}}\\
\label{kettedik}U(\lambda x) & \leq & \lambda^{\underline{\gamma}} U(x)+C\lambda^{\underline{\gamma}}.
\end{eqnarray}
\end{proposition}
\begin{proof}
Let $C:=\max(U(\overline{x}),-U(-\underline{x}))+c$.
Obviously, \eqref{egyedik} holds true for $x\geq \overline{x}$ by \eqref{ae+}.
For $0\leq x\leq \overline{x}$, as $U$ is nondecreasing, we get
\begin{eqnarray*}
U(\lambda x)\leq U(\lambda \overline{x})\leq \lambda^{\overline{\gamma}} U(\overline{x})+c,
\end{eqnarray*}
from \eqref{ae+} and \eqref{egyedik} holds true.
Now, for $-\underline{x} <x\leq0$,
$$\lambda^{\overline{\gamma}} U(x)+C\lambda^{\overline{\gamma}}  \geq \lambda^{\overline{\gamma}} U(-\underline{x})+C\lambda^{\overline{\gamma}}$$
and \eqref{egyedik} holds true since $C \geq -U(-\underline{x})$ and
$U(\lambda x) \leq 0$.

If $x \leq -\underline{x}$, $U(x)\leq 0$. By \eqref{ae-} and $\overline{\gamma}<\underline{\gamma}$, one has
\begin{eqnarray*}
U(\lambda x)\leq  \lambda^{\underline{\gamma}} U(x) \leq \lambda^{\overline{\gamma}} U(x) \leq \lambda^{\overline{\gamma}} U(x)
+\lambda^{\overline{\gamma}} C.
\end{eqnarray*}
We now turn to the proof of \eqref{kettedik}.
For $x>0$, using \eqref{egyedik}, $\overline{\gamma}<\underline{\gamma}$ and $U(x)\geq 0$:
\begin{eqnarray*}
U(\lambda x)\leq \lambda^{\overline{\gamma}} U(x)+C\lambda^{\overline{\gamma}} \leq \lambda^{\underline{\gamma}} U(x)
+C\lambda^{\underline{\gamma}}.
\end{eqnarray*}
For $-\underline{x} <x\leq0$
$$\lambda^{\underline{\gamma}} U(x)+C\lambda^{\underline{\gamma}} \geq \lambda^{\underline{\gamma}} U(-\underline{x})+
C\lambda^{\underline{\gamma}}
\geq 0\geq U(\lambda x),
$$
since $C \geq -U(-\underline{x})$.
Finally, \eqref{kettedik} for $x \leq -\underline{x} $ follows directly from
\eqref{ae-}.
\end{proof}

\begin{proposition}
\label{randomD}
Assume that $S$ satisfies the (NA) condition. Then, for all $t=1,\ldots,T$, $D_t$  satisfies Assumption \ref{hypD}.
\end{proposition}
\begin{proof}
By Proposition A.1 of \cite{RS05} (condition (NA) is not necessary at this point),
$D_t \in {\cal B}(\mathbb{R}^d) \otimes {\cal H}$ and for almost all $\omega$,
$D_t(\omega)$ is an affine subspace of $\mathbb{R}^d$. From g) of Theorem 3 in \cite{JS98}, under
condition (NA),
$D_t(\omega)$ is, in fact, a non-empty vector subspace of $\mathbb{R}^d$, for almost all $\omega$ since it contains $0$.
\end{proof}

\begin{proposition}\label{visszafele}
Assume that $S$ satisfies the (NA) condition and that Assumptions \ref{ae}  and \ref{bellmann} hold true.
One can choose versions of the random functions $U_t,\ 0\leq t\leq T$, which are almost surely
nondecreasing,
continuous, finite and satisfy, outside a fixed negligible set,
\begin{eqnarray}
\label{AE+U} U_t(\lambda x)& \leq &  \lambda^{\overline{\gamma}} U_t(x)+C\lambda^{\overline{\gamma}},\\
\label{AE-U}U_t(\lambda x) & \leq & \lambda^{\underline{\gamma}} U_t(x)+C\lambda^{\underline{\gamma}},
\end{eqnarray}
for all $\lambda \geq 1$ and $x \in \mathbb{R}$.
Moreover,
there exist $\mathcal{F}_{t-1}$-measurable, finite valued random variables $N_{t-1} >0$
such that:
\begin{eqnarray}
\label{limU}
P\left(U_t(-N_{t-1}) <-\frac{2C}{\kappa_{t-1}}-1\vert \mathcal{F}_{t-1}\right) \geq 1-\kappa_{t-1}/2,
\end{eqnarray}
here $C$ is the same constant as in \eqref{AE+U} and \eqref{AE-U} above and $\kappa_{t-1}$ is as in \eqref{valaki}.
Finally, there exist $\mathcal{F}_{t}\otimes\mathcal{B}(\mathbb{R})$-measurable
functions $\tilde{\xi}_{t+1}$, taking values in $D_{t+1}$, $0 \leq t\leq T-1$ such that, almost surely,
\begin{equation}\label{pompom}
\forall x\in\mathbb{R}
\quad U_t(\omega,x)=E(U_{t+1}(x+
\tilde{\xi}_{t+1}(x) \Delta S_{t+1})\vert\mathcal{F}_{t}).
\end{equation}
\end{proposition}
\begin{proof} Going backwards from $T$ to $0$, we will
apply Lemmata \ref{Kbound}, \ref{bound} and \ref{candymaxiessup} and Proposition
\ref{candy} with the choice
$V:=U_t,\ \mathcal{H}=\mathcal{F}_{t-1}, \mathcal{F}=\mathcal{F}_{t},\ D:=D_t,\
Y:=\Delta S_{t}$.
Then for each $x \in \mathbb{R}$, we will choose  the random function $U_{t-1}(x)$ to be $A(x)$ which is an
almost surely nondecreasing and continuous version of $U_{t-1}(x)$ (see Lemma \ref{candymaxiessup} and Remark \ref{vienna}).
So we need to verify that
Assumptions \ref{hypD}, \ref{AOAone}, \ref{hypV1},  \ref{hypVmoinsinfty}, \ref{hypV2} and \ref{hypV3} hold true.

We start by the ones which can be verified directly for all $t$.
The price process $S$ satisfies the (NA) condition.
So by
Proposition \ref{karakter}, Assumption
\ref{AOAone} holds true with $\alpha=\delta_{t-1}$ and $\beta=\kappa_{t-1}$. Moreover,
by Proposition \ref{randomD}, $D_t$ satisfies Assumption \ref{hypD}. Now
by Proposition \ref{propre1}, \eqref{ut-x} and \eqref{ut+x} hold true thus Lemma \ref{determ} with $V=U_{t}$, $Y=\Delta S_{t}$, $\mathcal{H}=\mathcal{F}_{t-1}$ implies that Assumption \ref{hypVmoinsinfty} holds true.

It remains to prove that Assumptions \ref{hypV1}, \ref{hypV2} and \ref{hypV3}
hold.
We start at time $t=T$.
The non-random function $U_T=U$ is continuous and non-decreasing by Assumption \ref{ae}, so Assumption \ref{hypV1} holds.
Equations \eqref{bajka} and \eqref{bajka-} for $V=U_T$ follow from Proposition \ref{egeszR}, so Assumption \ref{hypV2} (and also \eqref{AE+U} and
\eqref{AE-U} for $t=T$) holds.
Assumption \ref{hypV3} (and also \eqref{limU} for $t=T$) is satisfied because for any $x\geq \underline{x}$,
$$U(-x) \leq  \left(\frac{x}{\underline{x}}\right)^{\underline{\gamma}}U(-\underline{x})$$
by \eqref{ae-} and $U(-\underline{x}) < 0$ by \eqref{minus}, so we may choose $N_{T-1}:=\max\left(\underline{x}, \underline{x}
\left(\frac{-(2C/\kappa_{T-1})-2}{U(-\underline{x})}\right)^{\frac1{\underline{\gamma}}}\right)$.

Now we are able to use Proposition \ref{candy} and there exists a
function $\tilde{\xi}_{T}$ with values in $D_T$ such that \eqref{pompom} holds for $t=T-1$. Moreover, by Lemmata \ref{bound} and  \ref{candymaxiessup},
we can chose for $U_{T-1}(\omega, \cdot)$  an
almost surely nondecreasing (finite-valued) and continuous version  (namely $A(\omega, \cdot)$ see Lemma \ref{candymaxiessup} and Remark \ref{vienna}).
Hence Assumption \ref{hypV1} holds for $U_{T-1}$.
We now prove that Assumption \ref{hypV2} (and also \eqref{AE+U} and
\eqref{AE-U} for $t=T-1$) holds for $V=U_{T-1}$.
For some fixed $x \in \mathbb{R}$ and $\lambda \geq 1$,
almost surely
\begin{eqnarray*}
U_{T-1}(\lambda x) & = & E(U_T(\lambda x +\tilde{\xi}_{T}(\lambda x) \Delta S_T)
\vert\mathcal{F}_{T-1}) \\
& \leq &
\lambda^{\overline{\gamma}}(E(U_T(x+ (\tilde{\xi}_{T}(\lambda x)/
\lambda)  \Delta S_T)\vert\mathcal{F}_{T-1}) + C) \\
& \leq &  \lambda^{\overline{\gamma}} (U_{T-1}(x)+C).
\end{eqnarray*}
where the first inequality follows from \eqref{egyedik} for $U_T$ (or \eqref{AE+U} for $t=T$).
Clearly, there is a common zero-probability set outside which this holds for all rational $x,\lambda$.
Using continuity of $U_{T-1}$ just like in Lemma \ref{indis}, this extends to all $\lambda,x$.
Thus \eqref{AE+U} holds for $t=T-1$. By the same argument,  \eqref{AE-U} also holds for $t=T-1$. Thus Assumption \ref{hypV2} is proved for $V=U_{T-1}$.

It remains to show that Assumption \ref{hypV3} holds for $U_{T-1}$ (and also \eqref{limU} for $t=T-1$).
Choose $I_{T-1}=2C/\kappa_{T-1}+1$ which is a.s. finite-valued and invoke  Lemma \ref{bound} (with $V=U_T$) to get some non-negative, finite valued and
$\mathcal{F}_{T-1}$-measurable random variable $N'$ such that
$U_{T-1}(-N')\leq -I_{T-1}$ a.s.
Let us define the $\mathcal{F}_{T-2}$-measurable events
$$
A_m:=\{\omega: P(N' \leq m\vert
\mathcal{F}_{T-2})(\omega)\geq 1-\kappa_{t-2}(\omega)/2\},\ m\in\mathbb{N}.
$$
As  $P(N' \leq m\vert
\mathcal{F}_{T-2}) $ trivially tends to $1$  when $m \to \infty$, the union of the sets $A_m$  cover a full measure set hence,
after defining recursively the partition
$$
B_1:=A_1,\quad B_{m+1}:=A_{m+1}\setminus \left(\cup_{j=1}^m A_j\right),
$$
we can construct the non-negative, $\mathcal{F}_{T-2}$-measurable random variable
$$
N_{T-2}:=\sum_{m=1}^{\infty} m 1_{B_m}
$$
such that
$P(N' \leq N_{T-2}\vert \mathcal{F}_{T-2}) \geq
1-\kappa_{t-2}/2$ a.s. Then a.s. (recall that for a.e. $\omega$, $U_{T-1}(\omega,.)$ is non-decreasing):
\begin{eqnarray*}
P(U_{T-1}(-N_{T-2}) <-I_{T-1}\vert \mathcal{F}_{T-2})
 \geq  P(\{N' \leq N_{T-2}\} \cap
\{U_{T-1}(-N') <-I_{T-1}\} \vert \mathcal{F}_{T-2})
 \geq 1-\kappa_{T-2}/2.
\end{eqnarray*}
We are now able to use
Proposition \ref{candy} for $U_{T-1}$, \eqref{pompom} holds for $t=T-2$ and we can continue the procedure of dynamic programming
in an analogous way.
\end{proof}

\begin{proof}[of Theorem \ref{main}.]
We use the results of Proposition \ref{visszafele}. Set $\phi^*_1:=\tilde{\xi}_1(x)$ and define inductively:
$$
\phi^*_{t}:=
\tilde{\xi}_{t}\left(x+\sum_{j=1}^{t-1} \phi^*_j \Delta S_j\right) \
1\leq t\leq T.
$$
Joint measurability of $\tilde{\xi}_{t}$ assures that $\phi^*$ is a predictable process
with respect to the given
filtration. Lemma
\ref{candymaxiessup} and Propositions \ref{visszafele} and \ref{candy} (recall that we have chosen for $U_{t-1}$ in Proposition  \ref{visszafele} the good version $A$ of Lemma
\ref{candymaxiessup}) show that for $t=1,\ldots,T$ a.s.:
\begin{eqnarray}
\label{sarko}
E(U_{t}(V_{t}^{x,\phi^*})\vert\mathcal{F}_{t-1})=U_{t-1}(V_{t-1}^{x,\phi^*}).
\end{eqnarray}

We will now show that if $EU(V_T^{x,\phi^*})$ exists then $\phi^*\in\Phi(U,x)$ and for any strategy $\phi\in\Phi(U,x)$,
\begin{equation}\label{ok}
E(U(V^{x,\phi}_T))\leq E(U(V^{x,\phi^*}_T)).
\end{equation}
This will complete  the proof.

Let us consider first the case where $EU^+(V_T^{x,\phi^*})<\infty$.
Then by \eqref{sarko} and the
(conditional) Jensen inequality (see Corollary \ref{jensenjensen} with $g(x)=x^+$),
\begin{eqnarray*}
U_{T-1}^+(V_{T-1}^{x,\phi^*}) & \leq & E(U_T^+(V_T^{x,\phi^*})\vert\mathcal{F}_{T-1}) \mbox{ a.s.}
\end{eqnarray*}
Thus $E[U_{T-1}^+(V_{T-1}^{x,\phi^*})]<\infty$ and  repeating the argument, $E[U_{t}^+(V_{t}^{x,\phi^*})]<\infty$
for all $t$.

Now let us turn to the case where $EU^-(V_T^{x,\phi^*})<\infty$. The same argument as above with negative parts
instead of positive parts shows that $E[U_{t}^-(V_{t}^{x,\phi^*})]<\infty$, for all $t$.

It follows that, for all $t$,
$EU_{t}(V_t^{x,\phi^*})$ exists and so does $E(U_{t}(V_t^{x,\phi^*})\vert\mathcal{F}_{t-1})$ by Lemma \ref{ujdonatuj}.
This Lemma also implies that $E(E(U_{t}(V_t^{x,\phi^*})\vert\mathcal{F}_{t-1}))=EU_{t}(V_t^{x,\phi^*})$.
Hence
\begin{eqnarray}
\nonumber
E(U_T(V_T^{x,\phi^*}))& = & E(E(U_T(V_T^{x,\phi^*})\vert\mathcal{F}_{T-1}))=E(U_{T-1}
(V_{T-1}^{x,\phi^*}))\\
\label{nono}
& = & \ldots=E(U_0(x)).
\end{eqnarray}
By \eqref{lasagna} and \eqref{utetu},
$-\infty < U (x) \leq EU_0(x) < \infty$, hence also
$E(U_T(V_T^{x,\phi^*}))$ is finite and $\phi^*\in\Phi(U,x)$ follows.

Let $\phi\in\Phi(U,x)$, then $E(U(V^{x,\phi}_T))$ exists and is finite by definition of $\Phi(U,x)$.
By Lemma \ref{ujdonatuj}, we have that, for all $t$, $E(U(V^{x,\phi}_T)\vert\mathcal{F}_{t})$ exists
and that $E(E(U(V^{x,\phi}_T)\vert\mathcal{F}_{t}))=E(U(V^{x,\phi}_T))$.

We prove by induction that $E(U(V^{x,\phi}_T)\vert\mathcal{F}_{t}) \leq U_t(V^{x,\phi}_t)$ a.s.
For $t=T$, this is trivial.
Assume that it holds true for $t+1$.

Proposition \ref{propre1} (see \eqref{ut-x} and \eqref{ut+x}) and  Lemma \ref{determ} show that
$E(U^{\pm}_{t+1}(V^{x,\phi}_{t}+ \phi_{t+1}\Delta S_{t+1})\vert\mathcal{F}_{t}) <  +\infty$ and
$E(U_{t+1}(V^{x,\phi}_{t}+ \phi_{t+1}\Delta S_{t+1})\vert\mathcal{F}_{t})$ exists and it is finite.
So, by the induction hypothesis, \eqref{pompom}, Lemma
\ref{candymaxiessup}
and Proposition \ref{candy}, a.s.
$$
E(U(V^{x,\phi}_T)\vert\mathcal{F}_{t}) \leq E(U_{t+1}(V^{x,\phi}_{t}+ \phi_{t+1}\Delta S_{t+1})\vert\mathcal{F}_{t})\leq
E(U_{t+1}(V^{x,\phi}_{t}+ \tilde{\xi}_{t+1}(V_t^{x,\phi})\Delta S_{t+1})\vert\mathcal{F}_{t})=U_t(V_t^{x,\phi}).
$$

Applying the result at $t=0$, we obtain that $E(U(V^{x,\phi}_T)\vert\mathcal{F}_{0}) \leq U_0(x)$.
Using again $-\infty < U (x) \leq EU_0(x) < \infty$ (see \eqref{lasagna} and \eqref{utetu}), we obtain that
\begin{equation}\label{hasznos}
E(U(V^{x,\phi}_T))\leq E(U_0(x)).
\end{equation}
Putting \eqref{nono} and \eqref{hasznos} together, one gets
exactly \eqref{ok}.
\end{proof}

\begin{remark} We rectify here the statement of Theorem 2.7 in \cite{RS05}: just like in Theorem
 \ref{main} above, one has to add the condition that $EU(V^{c,\phi^*}_T)$ exists
 as this was implicitly assumed in its proof.
\end{remark}

We would like to check that Theorem \ref{main}
holds in a concrete, broad class of market models.
Let $\mathcal{M}$ denote the set of $\mathbb{R}$-valued random
variables $Y$ such that $E\vert Y\vert^p <\infty$ for all $p>0$.
This family is clearly closed under addition, multiplication and
taking conditional expectation. With a slight abuse of
notation, for a $d$-dimensional random variable $Y$, we write
$Y\in\mathcal{M}$ when we indeed mean $\vert Y\vert\in\mathcal{M}$.

\begin{proposition}
\label{camarche}
Let Assumption \ref{ae} hold and assume that,
\begin{equation}\label{errol}
U(x)\geq -m(|x|^p+1)\mbox{ for all }x\in\mathbb{R},
\end{equation}
holds for some $m,p>0$.
Furthermore, assume that for all
$0\leq t\leq T$ we have $\Delta S_t\in\mathcal{M}$ and
that (NA) holds with $\delta_t, \kappa_t$ of Proposition \ref{karakter}
satisfying $1/\delta_t,\ 1/\kappa_t\in\mathcal{M}$ for $0\leq t\leq T-1$.\\
Then there exists a solution $\phi^*$ of Problem \ref{leprobleme} with $\phi_t^*\in\mathcal{M}$ for $1\leq t\leq T$.
\end{proposition}
\begin{remark}
In the light of Proposition \ref{karakter}, $1/\delta_t,\ 1/\kappa_t\in\mathcal{M}$ for $0\leq t\leq T-1$
is a certain strong form of no-arbitrage. Note that if either $\kappa_t$ or $\delta_t$ is not constant,
then even a concave utility maximisation problem may be ill posed (see Example 3.3 in \cite{cr}), so
an integrability assumption on $1/\delta_t,\ 1/\kappa_t$ looks reasonable.

When $S$ has independent increments and (NA) holds, then one can choose $\kappa_t=\kappa$ and $\beta_t=\beta$
in Proposition \ref{karakter} with deterministic constants $\kappa,\beta>0$. These trivially
satisfy $1/\delta_t,\ 1/\kappa_t\in\mathcal{M}$ for $0\leq t\leq T-1$. See also section 8 of \cite{cr11} for
other concrete examples where $1/\delta_t,\ 1/\kappa_t\in\mathcal{M}$ is verified.

The assumption that $\Delta S_{t+1}, 1/\delta_t,\ 1/\kappa_t\in\mathcal{M} $ for $0\leq t\leq T-1$ could
be weakened to the existence of the $N$th moment for $N$ large
enough but this would lead to complicated book-keeping with no essential
gain in generality, which we
prefer to avoid.
\end{remark}
\begin{remark}
Assume that $U(x)\geq -m(|x|^p+1)$ holds true only for all $x \leq 0$. For $x \in \mathbb{R}$,
$
U(x)  =  U(x) 1_{x \leq 0}  + U(x)1_{x>0} \geq -m(|x|^p+1)1_{x \leq 0} + U(x) 1_{x>0}.
$
From Assumption \ref{ae}, $U(x) 1_{x>0} \geq U(0)=0$. Thus $U(x)\geq -m(|x|^p+1)$ holds true for all $x \in \mathbb{R}$ assuming only that it holds true for all $x \leq 0$.
\end{remark}

\begin{proof}[of Proposition \ref{camarche}.]
In order to prove Proposition \ref{camarche}, we need to refine the proof of Proposition \ref{visszafele}.
The price process $S$ satisfies the (NA) condition.
So by
Proposition \ref{karakter}, Assumption
\ref{AOAone} holds true with $\alpha=\delta_{t-1}$ and $\beta=\kappa_{t-1}$. Moreover,
by Proposition \ref{randomD}, $D_t$ satisfies Assumption \ref{hypD}.

\noindent\textsl{Claim : one can choose versions of the random function $U_t$ that satisfy Assumptions \ref{hypV1},  \ref{hypVmoinsinfty}, \ref{hypV2}
(with $\underline{\gamma}$ and $\overline{\gamma}$ defined in Assumption \ref{ae} and $C$ in Proposition \ref{egeszR}) and \ref{hypV3} (with
$\beta=\kappa_{t-1}$, $C$ defined in Proposition \ref{egeszR}, $N$ will be called $N_{t-1}$).
Moreover, $N_{t-1} \in {\cal M}$ and there exist non-negative, adapted random variables $C_t$, $J_{t-1}$, $M_{t-1}$ belonging
to ${\cal M}$ (i.e. $C_t$ is ${\cal F}_t$-measurable and $J_{t-1}$ and $M_{t-1}$ are ${\cal F}_{t-1}$-measurable) and numbers
$\lambda_t,\theta_{t-1}>0$ such that, for a.e. $\omega$,
\begin{eqnarray}
\label{ouf}
U_t(x) & \geq &  U(x), \mbox{ for all $x$}\\
\label{Uestborne}
U^+_t(x) & \leq &  C_t(|x|^{\lambda_t}+1),\mbox{ for all $x$},\\
\label{stratbornerec}
\tilde{K}_{t-1}(x) & \leq &  M_{t-1}(|x|^{\theta_{t-1}}+1)\; \mbox{ for all $x$}.
\end{eqnarray}
In addition, for all $x$, $y \in \mathbb{R}$,
\begin{eqnarray}
\label{espborne}
E(U^+_{t}(x+ |y| |\Delta S_{t}|)\vert\mathcal{F}_{t-1}) & \leq & J_{t-1} (|x|^{\lambda_{t}}+|y|^{\lambda_{t}}+1)< \infty, \mbox{ a.s.}
\end{eqnarray}
where the ${\cal F}_{t-1}$-measurable random variable $\tilde{K}_{t-1}(x)$ is just $\tilde{K}(x)$ defined
in \eqref{Ktilde} for the choice $V=U_{t}$, $Y=\Delta S_{t}$ and $\mathcal{H}:=\mathcal{F}_{t-1}$. Finally,
there exist $\mathcal{F}_{t-1}\otimes\mathcal{B}(\mathbb{R})$-measurable
functions $\tilde{\xi}_{t}$, taking values in $D_{t}$, such that, almost surely,
\begin{equation}\label{pompombis}
\forall x\in\mathbb{R}
\quad U_{t-1}(x)=E(U_{t}(x+
\tilde{\xi}_{t}(x) \Delta S_{t})\vert\mathcal{F}_{t-1}).
\end{equation}}

We proceed by backward induction starting at $t=T$.
By Assumption \ref{ae} and Proposition \ref{egeszR},
Assumptions \ref{hypV1} and \ref{hypV2} clearly hold.
Choosing
$$N_{T-1}:=\max \left(\underline{x},\underline{x}
\left(\frac{-(2C/\kappa_{T-1})-2}{U(-\underline{x})}\right)^{\frac1{\underline{\gamma}}}\right),$$ just like in the proof
of Proposition \ref{visszafele} (only \eqref{ae-} and \eqref{minus} from Assumption \ref{ae} were used there), we can see that Assumption \ref{hypV3}  holds true and $N_{T-1}\in\mathcal{M}$.

\eqref{ouf} is trivial and \eqref{veges} in Assumption \ref{hypVmoinsinfty} follows from \eqref{errol}.
We estimate, using Assumption \ref{ae} and the trivial $U(x)\leq U(\overline{x})$, $x\leq \overline{x}$,
\begin{eqnarray}\label{szalma}
U(x) & \leq & \frac{|x|^{\overline{\gamma}}}{\overline{x}^{\overline{\gamma}}}U(\overline{x})+ c +
U(\overline{x}) \leq C_T(|x|^{\overline{\gamma}}+1),
\end{eqnarray}
for all $x$, with $C_T=\max \left(\frac{U(\overline{x})}{\overline{x}^{\overline{\gamma}}}, c +
U(\overline{x}) \right)$. From Assumption \ref{ae}, $C_T$ is a non-negative constant and it is clear that
\eqref{szalma} also holds true for $U^+$ and thus \eqref{Uestborne} holds true  with $\lambda_T:=\overline{\gamma}$ (we are dealing with a deterministic function at this stage).
As $|x+y|^{\overline{\gamma}} \leq 2^{\overline{\gamma}}(|x|^{\overline{\gamma}}+|y|^{\overline{\gamma}})$, we obtain a.s.
\begin{eqnarray}
\nonumber E(U^+(x + |y||\Delta S_{T}|)\vert\mathcal{F}_{T-1})  & \leq & E(C_T\vert\mathcal{F}_{T-1})
(2^{\overline{\gamma}}|x|^{\overline{\gamma}}+1) + 2^{\overline{\gamma}}|y|^{\overline{\gamma}}
E(C_T|\Delta S_{T}|^{\overline{\gamma}}\vert\mathcal{F}_{T-1}) \\
\nonumber & \leq : & J_{T-1} (|x|^{\overline{\gamma}}+|y|^{\overline{\gamma}}+1) <\infty.
\end{eqnarray}
It is clear that $J_{T-1}$ belongs
to ${\cal M}$ (recall $\Delta S_{T} \in {\cal M}$)  and that $J_{T-1}$ is ${\cal F}_{T-1}$-measurable.
Thus \eqref{espborne}  and \eqref{nemvegtelen} hold true and  Assumption \ref{hypVmoinsinfty} is satisfied.
To finish with the step $t=T$, it remains to prove \eqref{stratbornerec}.
As \eqref{errol} holds true, we can use \eqref{stratborne} in Lemma \ref{Kbound} and we just have to prove that
$M=M_{T-1}\in\mathcal{M}$. From Lemma \ref{Kbound}, $M_{T-1}$ is a polynomial function of
${1}/{\delta_{{T-1}}}$, ${1}/{\kappa_{T-1}}$, $N_{T-1}$ and $L_{T}$, $L_t$ will be $L$ from Lemma
\ref{Kbound} corresponding to $V=U_t$. As $L_T=E(U_T^+(1+|\Delta S_T|)|\mathcal{F}_{T-1}) \leq 3 J_{T-1}$ we get that $L_T\in\mathcal{M}$ and $M_{T-1}\in\mathcal{M}$
as well (recall
that we assumed that ${1}/{\delta_{T-1}}$ and ${1}/{\kappa_{T-1}}$ belonged to ${\cal M}$). Now we are able to use Proposition \ref{candy} and there exists a
function $\tilde{\xi}_{T}$ with values in $D_T$ such that \eqref{pompombis} holds for $t=T-1$.


Let us now proceed to the step $t=T-1$. As Assumptions \ref{hypD}, \ref{AOAone},  \ref{hypV1}, \ref{hypVmoinsinfty}, \ref{hypV2} and \ref{hypV3}
hold true for $V=U_T$, we can apply Lemmata \ref{bound} and  \ref{candymaxiessup} for $V=U_{T}$, which shows that one can choose a version of $U_{T-1}$ which satisfies Assumption \ref{hypV1}.
Just like in the proof of Proposition \ref{visszafele},
Assumption \ref{hypV2} also holds true.
For $V=U_T$, we get that by Lemmata \ref{Kbound} and \ref{bound} for all $x$, a.s. (see \eqref{supborne}),
\begin{eqnarray}
\label{fin}
U_{T-1}(x) & \leq & E(U_{T}(|x|+\tilde{K}_{T-1} (x)|\Delta S_{T}|) \vert\mathcal{F}_{T-1}) \\
\nonumber
& \leq & E(U^+_{T}(|x|+\tilde{K}_{T-1} (x)|\Delta S_{T}|) \vert\mathcal{F}_{T-1}) \\
\nonumber
      & \leq &  E(C_{T}(||x|+\tilde{K}_{T-1} (x)|\Delta S_{T}||^{\lambda_T}+1)\vert\mathcal{F}_{T-1})\\
      \nonumber
& \leq: & C_{T-1} (|x|^{\max\{\lambda_T\theta_{T-1},\lambda_T\}}+1)
\end{eqnarray} for some positive ${\cal F}_{T-1}$-measurable $C_{T-1}$. Thus
one also gets that for all $x$, $U^+_{T-1}(x) \leq C_{T-1} (|x|^{\max\{\lambda_T\theta_{T-1},\lambda_T\}}+1)$ a.s.
As both $U^+_{T-1}$ and  $ x \to C_{T-1} (|x|^{\max\{\lambda_T\theta_{T-1},\lambda_T\}}+1)$ are continuous, $U^+_{T-1}(x) \leq C_{T-1} (|x|^{\max\{\lambda_T\theta_{T-1},\lambda_T\}}+1)$ holds for all $x$ simultaneously, outside
a fixed negligible set (see Lemma \ref{indis}) and
\eqref{Uestborne} is satisfied with $\lambda_{T-1}:=\max\{\lambda_T\theta_{T-1},\lambda_T\}$.
As $M_{T-1}$ and $C_{T}$ belong to $\mathcal{M}$ from step $t=T$, $C_{T-1}$ also belongs to ${\cal M}$.
Furthermore, for all $x$, $y$, a.s.
\begin{eqnarray*}
E(U_{T-1}^+(x + |y| |\Delta S_{T-1}|)\vert\mathcal{F}_{T-2})  & \leq & E(C_{T-1}\vert\mathcal{F}_{T-2})
(2^{\lambda_{T-1}}|x|^{\lambda_{T-1}}+1) + \\
&  & 2^{\lambda_{T-1}}|y|^{\lambda_{T-1}}
E(C_{T-1}|\Delta S_{T-1}|^{\lambda_{T-1}}\vert\mathcal{F}_{T-2}) \\
& \leq: & J_{T-2}(|x|^{\lambda_{T-1}}+|y|^{\lambda_{T-1}}+1) <\infty.
\end{eqnarray*}
As $J_{T-2}$ clearly belongs
to ${\cal M}$ and
$J_{T-2}$ is ${\cal F}_{T-2}$-measurable, \eqref{espborne} is proved.
So \eqref{nemvegtelen} in Assumption \ref{hypVmoinsinfty} holds true.

Choosing $\xi=0$ in \eqref{vanek}, we get by \eqref{ouf} for $t=T$ that, for all $x \in \mathbb{R}$,
$$U_{T-1}(x)\geq
E(U_{T}(x)|\mathcal{F}_{T-1}) \geq U(x)>-\infty\mbox{ a.s.}.$$
As both $U_{T-1},U$ are continuous, $U_{T-1}(x)\geq U(x)$ holds for all $x$ simultaneously, outside
a fixed negligible set (see Lemma \ref{indis}) and \eqref{ouf} holds true. \\
Thus, for all $x$, $y$, a.s.,
$U_{T-1}(x -|y| |\Delta S_{T-1}|)  \geq U(x -|y| |\Delta S_{T-1}|)$. This implies that
$$E(U^-_{T-1}(x -|y| |\Delta S_{T-1}|) \vert\mathcal{F}_{T-2}) \leq E(U^-(x -|y| |\Delta S_{T-1}|)\vert\mathcal{F}_{T-2})\leq m E(|x -|y| |\Delta S_{T-1}||^p+1)\vert\mathcal{F}_{T-2})<\infty,$$
by \eqref{errol}. Thus \eqref{veges} holds true and Assumption \ref{hypVmoinsinfty} follows.

We now establish the existence of $N_{T-2}\in\mathcal{M}$ such that Assumption \ref{hypV3} holds true with $N=N_{T-2}$ and
$V=U_{T-1}$. Let us take the random variable ${N}'$ constructed in the proof of Lemma \ref{bound} for $V=U_T$ which
is such that $U_{T-1}(-N')\leq -I_{T-1}$, where
$I_{T-1}:=(2C/\kappa_{T-1})+1$.
By \eqref{jeudi},
$N'$ is a polynomial function of ${1}/{\kappa_{T-1}}$,
$N_{T-1}$ (which belong to ${\cal M}$) and
$E(U^+_T(\bar{K}_{T-1} |\Delta S_{T}|) \vert \mathcal{F}_{T-1})$, where $\bar{K}_{T-1}$ is defined as
$\bar{K}$ (see \eqref{BARK}) when $V=U_T$. As  $\bar{K}_{T-1}$ is a polynomial function of
$N_{T-1}$, $1/\delta_{T-1}$, $1/\kappa_{T-1}$ and $L_T$, we have $\bar{K}_{T-1} \in {\cal M}$ (recall from the end of step $t=T$ that $L_{T} \in {\cal M}$).
As $E(U^+_T(\bar{K}_{T-1} |\Delta S_{T}|) \vert \mathcal{F}_{T-1})$ is bounded by $J_{T-1}(0+ \bar{K}_{T-1}^{\lambda_T}+1)$ by
\eqref{espborne} for $t=T$, we conclude that $N'$ belongs to $\mathcal{M}$. Let us now set
$$
N_{T-2}:=\frac{2E(N'\vert \mathcal{F}_{T-2})}{\kappa_{T-2}}\in\mathcal{M}.
$$
The (conditional) Markov inequality implies that a.s.
$$
P(N' > N_{T-2}\vert\mathcal{F}_{T-2})\leq \frac{E(N'\vert\mathcal{F}_{T-2})}{N_{T-2}}=\frac{\kappa_{T-2}}2.
$$
As in the proof of Proposition \ref{visszafele}, a.s.
\begin{eqnarray*}
P(U_{T-1}(-N_{T-2})\leq -I_{T-1}\vert\mathcal{F}_{T-2})& \geq  &   P(\{N' \leq N_{T-2}\} \cap
\{U_{T-1}(-N') <-I_{T-1}\} \vert \mathcal{F}_{T-2}) \\
& \geq  &   P(\{N' \leq N_{T-2}\}\vert\mathcal{F}_{T-2}) \geq 1-\kappa_{T-2}/2,
\end{eqnarray*}
showing Assumption \ref{hypV3} for $V=U_{T-1}$.

We now turn to \eqref{stratbornerec}.
From \eqref{errol} and \eqref{ouf}, one can apply \eqref{stratborne} in Lemma \ref{Kbound} and \eqref{stratbornerec} is satisfied with some $M_{T-2}$ which is a polynomial function of ${1}/{\delta_{{T-2}}}$, ${1}/{\kappa_{T-2}}$, $N_{T-2}$ and $L_{T-1}$. So we just have to prove that
$M_{T-2}\in\mathcal{M}$.  As $L_{T-1}=E(U_{T-1}^+(1+|\Delta S_{T-1}|)|\mathcal{F}_{T-2}) \leq 3 J_{T-2}$ we get that $L_{T-1}\in\mathcal{M}$ and $M_{T-2}\in\mathcal{M}$
as well.
This concludes the step $t=T-1$. We are able to use Proposition \ref{candy} and there exists a
function $\tilde{\xi}_{T-1}$ with values in $D_{T-1}$ such that \eqref{pompombis} holds for $t=T-2$ and
one can continue this inductive procedure in an analogous way. The claim is proved.

Now, since by \eqref{Uestborne}
$$
EU_0(x)\leq EU^+_0(x)\leq (|x|^{\lambda_0} +1)EC_0<\infty,
$$
\eqref{lasagna} holds true and thus Assumption \ref{bellmann} is satisfied.

Set $\phi^*_1:=\tilde{\xi}_1(x)$ and define inductively:
$$
\phi^*_{t}:=
\tilde{\xi}_{t}\left(x+\sum_{j=1}^{t-1} \phi^*_j \Delta S_j\right) \
1\leq t\leq T.
$$
As in the proof of Theorem \ref{main}, joint measurability of $\tilde{\xi}_{t}$ assures that $\phi^*$ is a predictable process
with respect to the given
filtration.
We set $V_{t}^{x,\phi^*}= x+\sum_{j=1}^{t} \phi^*_j \Delta S_j$.
We show by induction that $\phi_t^*\in\mathcal{M}$ (and thus $\phi^*\in\Phi(U,x)$) and
$V_t^{x,\phi^*}\in\mathcal{M}$ for all $t$.

First, by \eqref{macha} and \eqref{stratbornerec}, on a full measure set, $\forall x\in \mathbb{R}$, $\forall 0 \leq t \leq T$, we get that
\begin{eqnarray}
\label{lancelot}
|\tilde{\xi}_t(x) |\leq \tilde{K}_{t-1}(x) \leq M_{t-1} (1+\vert x\vert^{\theta_{t-1}}),
\end{eqnarray}
where $M_{t-1} \in \mathcal{M}$.\\
For $t=1$, as $\phi^*_1=\tilde{\xi}_1(x)$, \eqref{lancelot} shows  that
$\phi_1^*\in\mathcal{M}$. This implies that $V_{1}^{x,\phi^*}= x+ \phi^*_1 \Delta S_1\in\mathcal{M}$.\\
Assume that for some $t$, $\phi_{t-1}^*\in\mathcal{M}$ and
$V_{t-1}^{x,\phi^*}\in\mathcal{M}$.
By \eqref{lancelot} again,
$$|\phi^*_{t}|=
\left|\tilde{\xi}_{t}\left(V_{t-1}^{x,\phi^*}\right)\right| \leq M_{t-1} (1+\vert V_{t-1}^{x,\phi^*}\vert^{\theta_{t-1}}),$$
and thus $\phi^*_{t} \in\mathcal{M}$.
As $V_{t}^{x,\phi^*}= V_{t-1}^{x,\phi^*}+ \phi^*_t \Delta S_t$, we also get that
$V_{t}^{x,\phi^*}\in\mathcal{M}$ and the argument is complete.

Now by \eqref{errol} and \eqref{ouf}, $U_{t}(V_{t}^{x,\phi^*}) \geq U(V_{t}^{x,\phi^*})\geq -m(|V_{t}^{x,\phi^*}|^p+1)$.
Using  \eqref{Uestborne},
$U_{t}(V_{t}^{x,\phi^*}) \leq   C_t(|V_{t}^{x,\phi^*}|^{\lambda_t}+1)$ and thus $U_{t}(V_{t}^{x,\phi^*}) \in \mathcal{M}$. In particular
$E(U_{t}(V_{t}^{x,\phi^*}))$ and $E(U_0(x))$ are finite.

Recall that from Lemma \ref{candymaxiessup}, Propositions \ref{visszafele} and \ref{candy}, for $t=1,\ldots,T$, one has
$$
E(U_{t}(V_{t}^{x,\phi^*})\vert\mathcal{F}_{t-1})=U_{t-1}(V_{t-1}^{x,\phi^*}) \mbox{ a.s.}
$$
Thus
\begin{eqnarray}
\nonumber
E(U_T(V_T^{x,\phi^*}))& = & E(E(U_T(V_T^{x,\phi^*})\vert\mathcal{F}_{T-1}))=E(U_{T-1}
(V_{T-1}^{x,\phi^*}))\\
\label{nono1}
& = & \ldots=E(U_0(x)).
\end{eqnarray}
As in the proof of Theorem \ref{main}, for any $\phi\in\Phi(U,x)$,  we obtain that $E(U(V^{x,\phi}_T)\vert\mathcal{F}_{0}) \leq U_0(x)$ a.s.
As $EU_0(x)<\infty$, it follows that
$
E(U(V^{x,\phi}_T))\leq E(U_0(x)).$
So from \eqref{nono1}, one gets
$$
E(U(V^{x,\phi}_T))\leq E(U(V^{x,\phi^*}_T)).
$$
for all $\phi\in\Phi(U,x)$.
This completes the proof.
\end{proof}

We provide one more result in the spirit of Proposition \ref{camarche}.
\begin{proposition}
\label{camarche1}
Let Assumption \ref{ae} hold and let $\Delta S_t$,
$0\leq t\leq T$ be a bounded process. Let (NA) hold with $\delta_t, \kappa_t$ of Proposition \ref{karakter}
being constant.
Then there exists a solution $\phi^*\in\Phi(U,x)$ of Problem \ref{leprobleme} which is a bounded process.
\end{proposition}
\begin{proof} In this case we note that
$$
U(x)\geq -U^-(x),\mbox{ for all }x\in\mathbb{R}
$$
holds instead of \eqref{errol} and $U^-$ is a continuous, hence also locally bounded non-negative function.
Thus in Lemmata \ref{Kbound} and \ref{bound}, assuming that $V(x)\geq -U^-(x)$ a.s. for all $x\in\mathbb{R}$, we obtain that
$\tilde{K}(x)$ (see \eqref{Ktilde}) is a locally bounded function of $x,N,1/\alpha,1/\beta,L$ and $U^-(-\lfloor x\rfloor^-)$ and $\bar{K}$ (see \eqref{BARK}) is a polynomial  function of $N,1/\alpha,1/\beta$ and $L$.
So one can imitate the proof of Proposition \ref{camarche} and get that the $\tilde{\xi}_t(\cdot)$
are also locally bounded. Hence the $V^{x,\phi^*}_{t-1}$ and $\phi^*_t$  as well  and we can conclude.
\end{proof}

\section{Conclusions}\label{se5}

One may try to prove a result similar to Theorem \ref{main} in continuous-time models. In the light of
results in \cite{jz}, however, serious limitations are encountered soon. In \cite{jz} the authors consider a setting where investors maximise
a functional possibly involving distorted probabilities. If we look at the particular case of no distortion (which is the setting
of our present paper), Theorem 3.2 of \cite{jz} implies
that taking $U(x)=x^{\alpha}$, $x>0$ and $U(x)=-(-x)^{\beta}$, $x\leq 0$ with $0<\alpha,\beta\leq 1$
the utility maximisation problem becomes ill-posed even in the simplest Black and Scholes model (in the presence of distortions the problem may be
well-posed).

On one hand, this
shows that there is a fairly limited scope for the extension of our results to continuous-time market models unless the set
of strategies is severely restricted (as in \cite{bkp}, \cite{cp} and \cite{cd}). On the other hand, this underlines
the versatility and power of discrete-time modeling. The advantageous properties present in the discrete-time setting do not always carry over to the continuous-time case
which is only an idealization of the real trading mechanism.

\section{Appendix}\label{qpp}

\subsection{Generalized conditional expectation}
Let
$W$ be a non-negative random variable on the probability space $(\Omega,\Im,P)$.
Let $\mathcal{H}\subset\Im$ be a sigma-algebra. Define (as in e.g. \cite{dm}), the generalized
conditional expectation by
\[
E(W\vert\mathcal{H}):=\lim_{n\to\infty}E(W\wedge n\vert\mathcal{H}),
\]
where the limit a.s. exists by monotonicity (but may be $+\infty$). In particular,
$EW$ is defined (finite or infinite). Note that if $EW<+\infty$, then the
generalized and the usual conditional expectations of $W$ coincide.

\begin{lemma}\label{ohh}
For all $A\in\mathcal{H}$ and all non-negative random variables $W$, the following equalities hold a.s.:
\begin{eqnarray}
\label{wax}
E(1_A E(W\vert\mathcal{H})) &=& E(W1_A)\\
\label{waxx}
E(W1_A\vert\mathcal{H}) & = & E(W\vert\mathcal{H})1_A.
\end{eqnarray}

Furthermore, $E(W\vert\mathcal{H})<+\infty$ a.s. if and only if there is a sequence $A_m\in\mathcal{H}$,
$m\in\mathbb{N}$ such that $E(W1_{A_m})<\infty$ for all $m$ and $\cup_m A_m=\Omega$.
In this case,
$E(W\vert\mathcal{H})$ is the Radon-Nykodim derivative of the sigma-finite measure
$\mu(A):=E(W1_A)$, $A\in\mathcal{H}$ with respect to $P$ on $(\Omega,{\mathcal{H}})$.
\end{lemma}
\begin{proof} Most of these facts are stated in section II.39 on page 33 of \cite{dm}.
We nevertheless give a quick proof for the sake of completeness. Let $A\in\mathcal{H}$ arbitrary. Then
\begin{eqnarray*}
E(1_A E(W\vert\mathcal{H}))& = & \lim_{n\to\infty} E(1_A E(W\wedge n\vert\mathcal{H})) \\
& = & \lim_{n\to\infty} E((W\wedge n)1_A) = E(W1_A)
\end{eqnarray*}
by monotone convergence and by the properties of ordinary conditional expectations.
Similarly, \eqref{waxx} is satisfied by monotone convergence and
by the properties of ordinary conditional expectations.

Now, if $A_m$ is a sequence as in the statement of Lemma \ref{ohh}, then $\mu$ is indeed sigma-finite and \eqref{wax} implies that $E(W\vert\mathcal{H})$
is the Radon-Nykodim derivative of $\mu$ with respect to $P$ on $(\Omega,\mathcal{H})$ and as
such, it is a.s. finite.

Conversely, if $E(W\vert\mathcal{H})<+\infty$ a.s. then define
$A_m:=\{E(W\vert\mathcal{H})\leq m\}$. We have, by \eqref{wax},
$$
E(W1_{A_m})=E(1_{A_m}E(W\vert\mathcal{H}))\leq m<\infty,
$$
showing the existence of a suitable sequence $A_m$.
\end{proof}

For a real-valued random variable $Z$ we may define,
if either $E(Z^+\vert\mathcal{H})<\infty$ a.s. or $E(Z^-\vert\mathcal{H})<\infty$ a.s.,
\[
E(Z\vert\mathcal{H}):=E(Z^+\vert\mathcal{H})-E(Z^-\vert\mathcal{H}).
\]
In particular, $E(Z)$ is defined if either $E(Z^+)<+\infty$ or $E(Z^-)<+\infty$.

\begin{lemma}\label{ujdonatuj}
If $E(Z)$ is defined then so is $E(Z\vert\mathcal{H})$ a.s. and $E(Z)=E(E(Z\vert\mathcal{H}))$.
\end{lemma}
\begin{proof}
We may suppose that e.g. $E(Z^+)<\infty$. Then $E(Z^+\vert\mathcal{H})$
exists (in the ordinary sense as well) and is finite, so $E(Z\vert\mathcal{H})$
exists a.s. Then, by \eqref{wax}, we have $E(Z^{\pm})=E(E(Z^{\pm}\vert\mathcal{H}))$.
\end{proof}

\begin{corollary}\label{nagyy}
Let $Z$ be a random variable and let $W$ be an  $\mathcal{H}$-measurable random variable.
Assume that there is a sequence $A_m\in\mathcal{H}$,
$m\in\mathbb{N}$ such that $\cup_m A_m=\Omega$ and $E(Z1_{A_m} |\mathcal{H})$ exists
and it is finite a.s. for all $m$. Then

(i) $E(Z|\mathcal{H})$ exists and it is finite a.s.

(ii) If $W1_{A_m}\leq E(Z1_{A_m} |\mathcal{H})$ a.s. for all $m$ then
$W\leq E(Z|\mathcal{H})$ a.s.

(iii) If $W1_{A_m}=E(Z1_{A_m} |\mathcal{H})$ a.s. for all $m$
then $W= E(Z|\mathcal{H})$ a.s.

\noindent This corollary applies, in particular, when $E(Z\vert\mathcal{H})$ is known to exist and to be finite a.s.
\end{corollary}
\begin{remark}
As in $(ii)$ or $(iii)$, $W1_{A_m}\leq E(Z1_{A_m} |\mathcal{H})$ a.s. and  $E(Z1_{A_m} |\mathcal{H}) < \infty$ a.s. for all $m$, $W$ is necessarily finite on $\cap_{m\in\mathbb{N}}\{W1_{A_m}< \infty\}$ which is of full measure.
\end{remark}
\begin{proof}[of Corollary \ref{nagyy}.]
Fix some $m$ such that $E(Z1_{A_m} |\mathcal{H})$ exists
and it is finite a.s., then $E(|Z|1_{A_m} |\mathcal{H})$ is also finite a.s. and by Lemma \ref{ohh}
there exists  a sequence $(B^{m}_j)_j$ such that $\cup_j B^{m}_j=\Omega$ and
$E(|Z|1_{A_m}1_{B_j^{m}})<\infty$ for all $j$.

Then the sets $C(m,j):=A_m\cap B^{m}_j$ are such that
$\cup_{m,j}C(m,j)=\Omega$. Let $C_n,n\in\mathbb{N}$ be the enumeration of all
the sets $C(m,j)$.  We clearly have
$E(|Z| 1_{C_n})<\infty$ for all $n$. Hence, by Lemma \ref{ohh}, $E(|Z||\mathcal{H})<\infty$ and
thus $E(Z|\mathcal{H})$ exists and is finite a.s.

Suppose that, e.g., $\{W>E(Z\vert\mathcal{H})\}$ on a set of positive measure.
Then there is $n$ such that $G:=C_n\cap \{W>E(Z\vert\mathcal{H})\}$ has positive
measure. There is also $m$ such that $C_n\subset A_m$. Then
$E(|Z|1_G)\leq E(|Z|1_{C_n})<\infty$ and
\begin{eqnarray*}
E(E(Z\vert\mathcal{H})1_G)  &=& E(E(Z1_{A_m}\vert\mathcal{H})1_G)\geq
E(W1_{A_m}1_G)=E(W1_G),
\end{eqnarray*}
but this contradicts the choice of $G$, showing $W\leq E(Z|\mathcal{H})$ a.s.
Arguing similarly for $\{W<E(Z\vert\mathcal{H})\}$ we can get (iii)
as well.
\end{proof}

\begin{lemma}\label{lebbencs}
Let $Z_n$ be a sequence of random variables with $|Z_n|\leq W$ a.s., $n\in\mathbb{N}$ converging to $Z$ a.s. If $E(W\vert\mathcal{H})<\infty$
a.s. then $E(Z_n\vert\mathcal{H})\to E(Z\vert\mathcal{H})$ a.s.
\end{lemma}
\begin{proof}
Let $A_m\in\mathcal{H}$ be a partition of $\Omega$ such that $E(W1_{A_m})<\infty$ for all $m$.
Fixing $m$, the statement follows on $A_m$ by the ordinary conditional Lebesgue theorem. Since the $A_m$ form a partition, it holds
a.s. on $\Omega$.
\end{proof}

\begin{corollary}\label{jensenjensen}
Let $g:\mathbb{R}\to\mathbb{R}$ be convex and bounded from below. Let $E(Z\vert\mathcal{H})$
exist and be finite a.s. Then
$$
E(g(Z)\vert\mathcal{H})\geq g(E(Z\vert\mathcal{H}))\mbox{ a.s.}
$$
\end{corollary}
\begin{proof} We may and will assume $g(0)=0$. Define
$B:=\{E(g(Z)\vert\mathcal{H})<\infty\}$. The inequality is trivial on the complement of $B$.

As $E(|Z|\vert\mathcal{H})<\infty$ a.s. and $E( |g(Z)| 1_B \vert \mathcal{H})<\infty$ a.s. (recall that $g$ is bounded from below),
from Lemma \ref{ohh}, one can find a sequence $A_m$ such that $\cup_m A_m=\Omega$ and both
$E(|Z|1_{A_m})<\infty$ and $E( |g(Z)| 1_{A_m}1_B)<\infty$ hold true for all $m$. From the ordinary (conditional) Jensen inequality we clearly have
$$
1_B E(g(Z)1_{A_m}|\mathcal{H})=E(g(Z1_{A_m}1_B)\vert\mathcal{H})
\geq g(E(Z1_{A_m}1_B\vert\mathcal{H}))=g(E(Z\vert\mathcal{H}))1_{A_m}1_B, \mbox{ a.s.}
$$
for all $m$,
and the statement follows if we can apply Corollary \ref{nagyy}, i.e. if $E(g(Z) 1_{A_m}\vert\mathcal{H})$
exists and it is finite a.s. This holds true by the choice of $A_m$.
\end{proof}

\subsection{Further useful results}

We start with a simple but useful Lemma.
\begin{lemma}\label{indis}
Let $(\Omega, {\cal H},P)$ a probability space.
Let $U$ and $V$ from $\Omega \times \mathbb{R} $ to $\mathbb{R}$ such that for all $x \in \mathbb{R}$,
$U(\cdot,x), V(\cdot, x)$ are ${\cal H}$-measurable.
Assume that for a.e. $\omega$, $U(\omega,\cdot)$ and  $V(\omega, \cdot)$ are either both right-continuous or
both left-continuous. \\
(i) If for all $q \in \mathbb{Q}$, $U(\cdot,q) \leq V(\cdot,q)$ a.s.
then a.s., $U(\cdot,x) \leq V(\cdot,x), \mbox{ for all } x \in \mathbb{R}.$\\
(ii) If for all $q \in \mathbb{Q}$, $U(\cdot,q) = V(\cdot,q)$ a.s.
then a.s., $U(\cdot,x) = V(\cdot,x), \mbox{ for all } x \in \mathbb{R}.$
\end{lemma}
\begin{proof}
Assume that $U$ and $V$ are a.e. left-continuous and let us prove (i) (the proof of (ii) is similar).
We denote by $$\bar{\Omega}=\{\omega\,|\, U(\cdot,\omega) \mbox{ is left-continuous}\}\cap \{\omega\,|\, V(\omega,\cdot) \mbox{ is left-continuous}\}
\cap \left(\cap_{q \in \mathbb{Q}}\{U(\cdot,q) \leq V(\cdot,q)\}\right).$$
Clearly $P(\bar{\Omega})=1$. Let $\omega \in \bar{\Omega}$. Let $x \in \mathbb{R}$. There exists $(q_p)_p \subset \mathbb{Q}$ such that $q_p\nearrow x$. Then, by definition of
$\bar{\Omega}$, $U(\omega,q_p) \to U(\omega,x)$ and $V(\omega,q_p) \to V(\omega,x)$. As $U(\omega,q_p) \leq V(\omega,q_p)$ again by definition of
$\bar{\Omega}$, we get that $U(\omega,x) \leq V(\omega,x)$ and the result is proved.
\end{proof}

\begin{lemma}
\label{gast}
Let $(\Omega, {\cal H},P)$ be a complete probability space.
Let $F:\Omega\times\mathbb{R}^d\to\mathbb{R}$ be a function such that for almost all $\omega\in\Omega$,
$F(\omega,\cdot)$ is continuous and for each $y\in\mathbb{R}^d$, $F(\cdot,y)$ is ${\cal H}$-measurable. Let $K>0$ be
an ${\cal H} $-measurable random variable.

Set $f(\omega)=\mathrm{ess.}\sup_{\xi\in \Xi, |\xi|\leq  K} F(\omega,\xi(\omega))$.  Then,
for almost all $\omega$,
\begin{eqnarray}
\label{lavieestbienfaite2}
f(\omega) & = & \sup_{y \in \mathbb{R}^d, |y| \leq K(\omega) } F(\omega,y).
\end{eqnarray}
\end{lemma}
\begin{proof}
By p. 70 of \cite{CV77}, $F$ is $\mathcal{H}\otimes\mathcal{B}(\mathbb{R}^d)$-measurable and so is
$$
\sup_{y \in \mathbb{R}^d, |y| \leq K(\omega) } F(\omega,y)=\sup_{y \in \mathbb{Q}^d, |y| \leq K(\omega) } F(\omega,y).
$$
Hence $\sup_{y \in \mathbb{R}^d, |y| \leq K(\omega) } F(\omega,y) \geq  f(\omega)$ a.s.
by the definition of essential supremum.
Assume that the inequality is strict with positive probability. Then for some $\varepsilon>0$
the set
$$A=\{(\omega,y)\in \Omega \times  \mathbb{R}^d \ : \ |y| \leq K(\omega); \ F(\omega,y) -f(\omega)
\geq \varepsilon \}$$
has a projection $A'$ on $\Omega$ with $P(A')>0$.
Recall that $\omega \rightarrow F(\omega,\xi(\omega))$ is ${\cal H} $-measurable for $\xi\in \Xi$.
By definition of the essential supremum, $f$ is ${\cal H} $-measurable and
hence $A \in {\cal H} \otimes {\cal B}(\mathbb{R}^d)$. The measurable selection theorem (see for example
Proposition III.44 in \cite{dm}) applies and
there exists some  ${\cal H} $-measurable random variable $\eta$ such that $(\omega, \eta(\omega))\in A$
for $\omega\in A'$ (and  $\eta(\omega)=0$ on the complement of $A'$). This leads to a contradiction since for all $\omega \in A'$,
$f(\omega)<F(\omega,\eta(\omega))$ by the construction of $\eta$ and
$f(\omega)\geq
F(\omega,\eta(\omega))$ a.s. by the definition of $f$.
\end{proof}

\section*{Acknowledgments.} Part of this work was carried out while
the first author was affiliated to the University of Paris 7 and the second one to the University of Edinburgh. The first author thanks LPMA (UMR 7599) for support.
The authors thank an anonymous referee for his/her detailed and helpful comments.
The second author thanks Teemu Pennanen for drawing his attention to several important references.

\end{document}